\documentclass[aps,prd,onecolumn,showpacs,groupedaddress,nofootinbib]{revtex4-2}
\usepackage{graphicx}
\usepackage{dcolumn}
\usepackage{bm}
\usepackage{color}
\usepackage{longtable}
\usepackage{supertabular}
\usepackage[T1]{fontenc}
\usepackage{epstopdf}
\usepackage{amsmath}
\usepackage{amssymb}
\usepackage{setspace}
\usepackage{multirow}
\usepackage{booktabs}
\usepackage[section]{placeins}
\usepackage{mathrsfs}
\usepackage{hyperref}

\usepackage{amsthm}

 \newcommand{\bq}{\begin{equation}}
 \newcommand{\eq}{\end{equation}}
 \newcommand{\bqn}{\begin{eqnarray}}
 \newcommand{\eqn}{\end{eqnarray}}
 \newcommand{\nb}{\nonumber}

\newtheorem{definition}{\textbf{Definition}}[section]
\newtheorem{Lem}[definition]{\textbf{Lemma}}
\newtheorem{Thm}[definition]{\textbf{Theorem}}

\newtheorem{Cor}[definition]{\textbf{Corollary}}

\begin{document}

\newcommand{\markerone}{\raisebox{0pt}{\tikz{\node[draw,scale=0.3,regular polygon, regular polygon sides=3,fill=red!1!red,rotate=0](){};}}}
\newcommand{\markertwo}{\raisebox{0.5pt}{\tikz{\node[draw,scale=0.3,regular polygon, regular polygon sides=3,fill=blue!1!blue,rotate=180](){};}}}

\title{Capacity drop induced by stability modifications in stochastic dynamic systems}

\author{Mariana Pereira de Melo$^{1}$}\email[E-mail: ]{marianapm@usp.br}
\author{Leon Alexander Valencia$^{2}$}\email[E-mail: ]{lalexander.valencia@udea.edu.co}
\author{Wei-Liang Qian$^{1,3,4}$}\email[E-mail: ]{wlqian@usp.br (corresponding author)}

\affiliation{$^{1}$ Escola de Engenharia de Lorena, Universidade de S\~ao Paulo, 12602-810, Lorena, SP, Brazil}
\affiliation{$^{2}$ Instituto de Matemáticas, Universidad de Antioquia, Calle 67 Nro. 53-108, Medellín, Colombia}
\affiliation{$^{3}$ Faculdade de Engenharia de Guaratinguet\'a, Universidade Estadual Paulista, 12516-410, Guaratinguet\'a, SP, Brazil}
\affiliation{$^{4}$ Center for Gravitation and Cosmology, School of Physical Science and Technology, Yangzhou University, 225002, Yangzhou, Jiangsu, China}

\date{Mar. 14th, 2025}

\begin{abstract}
It is well-known that the fundamental diagram in a realistic traffic system is featured by capacity drop. 
From a mesoscopic approach, we demonstrate that such a phenomenon is linked to the unique properties of stochastic noise, which, when considered a specific perturbation, may counterintuitively enhance the stability of the originally deterministic system.
We argue that the system achieves its asymptotic behavior through a trade-off between the relaxation towards the stable attractors of the underlying deterministic system and the stochastic perturbations that non-trivially affect such a process.
Utilizing It\^o calculus, the present study analyzes the threshold of the relevant perturbations that appropriately give rise to such a physical picture.
In particular, we devise a scenario for which the stochastic noise is introduced in a minimized fashion to a deterministic fold model, which is known to reproduce the main feature of the fundamental diagram successfully.
Our results show that the sudden increase in vehicle flow variance and the onset of capacity drop are intrinsically triggered by stochastic noise.
Somewhat counterintuitively, we point out that the prolongation of the free flow state's asymptotic stability that forms the inverse lambda shape in the fundamental diagram can also be attributed to the random process, which typically destabilizes the underlying system.
The asymptotic behaviors of the system are examined, corroborating well with the main features observed in the experimental data.
The implications of the present findings are also addressed.
\end{abstract}

\maketitle
\newpage

\section{Introduction}\label{sec1}

Over the past decades, the study of traffic flow dynamics has seen a surge in interest, stemming from a demand to understand and mitigate the challenges of growing urbanization and motorization. 
The study on the subject dates back to the 1930s, with notable reference to Greenshield's analysis on traffic capacity~\cite{traffic-flow-data-03}.
The approach pioneered by Lighthill, Whitham~\cite{traffic-flow-hydrodynamics-01}, and Richards~\cite{traffic-flow-hydrodynamics-02} (LWR) marks the first attempt to depict the traffic flow along a road segment as a continuum.
It effectively treats a system consisting of discrete vehicles by a continuous distribution of matter, whose dynamics are governed by deterministic hydrodynamic equations.

Since the inception of this field, the research community has substantially diversified.
A viable perspective for distinguishing different theoretical frameworks is to focus on the level of data aggregation corresponding to the scale at which traffic data is collected.
This spectrum spans from macroscopic, which focuses on large-scale flow characteristics, to mesoscopic, which bridges the gap between macro and micro scales by focusing on the accumulations of vehicles in the underlying phase space, and finally to microscopic, which captures detailed individual vehicle dynamics.
Macroscopic models~\cite{traffic-flow-hydrodynamics-01, traffic-flow-hydrodynamics-02, traffic-flow-hydrodynamics-11, traffic-flow-hydrodynamics-12, traffic-flow-hydrodynamics-13, traffic-flow-hydrodynamics-14, traffic-flow-hydrodynamics-17, traffic-flow-hydrodynamics-22, traffic-flow-hydrodynamics-23, traffic-flow-hydrodynamics-24, traffic-flow-hydrodynamics-25, traffic-flow-hydrodynamics-30, traffic-flow-hydrodynamics-36, traffic-flow-hydrodynamics-37, traffic-flow-hydrodynamics-40} primarily follow the approach initiated by the LWR model and treat traffic flow as a fluid.
Mesoscopic approaches~\cite{traffic-flow-btz-01, traffic-flow-btz-02, traffic-flow-btz-03, traffic-flow-btz-05, traffic-flow-btz-14, traffic-flow-btz-15, traffic-flow-btz-lob-02}, typically presented in the form of the Boltzmann equation, consider the system's evolution in phase space, interpreting traffic as an interacting ``gas'' of particles.
Microscopic schemes~\cite{traffic-flow-micro-15, traffic-flow-micro-18, traffic-flow-micro-19, traffic-flow-micro-32, traffic-flow-micro-33, traffic-flow-micro-37, traffic-flow-micro-51}, often implemented through car-following models, emphasize the behaviors and interactions of individual vehicles, including cellular automaton approaches~\cite{traffic-flow-cellular-01, traffic-flow-cellular-02, traffic-flow-cellular-03, traffic-flow-cellular-04} defined on a discrete space-time grid.
Due to the distinct nature of these theoretical tools, each approach offers unique insights into its underlying physics.
In this regard, the above approaches have been utilized to scrutinize various realistic aspects of traffic flow, such as the human driver behavior analysis regarding lane-changing and gap acceptance~\cite{traffic-flow-phenomenology-12, traffic-flow-micro-41, traffic-flow-micro-62, traffic-flow-micro-63}, intricate traffic networks that account for urban road infrastructure~\cite{traffic-flow-phenomenology-40, traffic-flow-phenomenology-41, traffic-flow-phenomenology-42}, and the optimization of traffic signals~\cite{traffic-flow-micro-78, traffic-flow-micro-81}.
Also, much attention has been concentrated on sustainable mobility and urban traffic integration, aiming at a transportation system that is both efficient and environmentally and socially sustainable~\cite{traffic-flow-phenomenology-43, traffic-flow-phenomenology-48}. 
Studies have been performed regarding macroscopic traffic patterns in an attempt to understand the traffic system at an aggregate level across the transportation network. 
By examining the overall traffic flow characteristics, one aims to explain how these factors interact over entire networks rather than at the individual vehicle level~\cite{traffic-flow-phenomenology-25, traffic-flow-review-20, traffic-flow-phenomenology-46}.
Furthermore, data-driven analysis has come to play an increasingly important role~\cite{traffic-flow-micro-57, traffic-flow-micro-80}, and in particular, machine learning algorithms have manifestly demonstrated their potential~\cite{traffic-flow-data-ML-15, traffic-flow-data-ML-20, traffic-flow-data-ML-21, traffic-flow-data-ML-27, traffic-flow-data-ML-28, traffic-flow-data-ML-29, traffic-flow-data-ML-30}.

A central topic that remains relevant in studying traffic flow is the instability of the underlying dynamical system that is interpreted to trigger a congested traffic state.
Since it was first pointed out in seminal works by Bando {\it et al.}~\cite{traffic-flow-micro-15} and Kerner {\it et al.}~\cite{traffic-flow-hydrodynamics-14, traffic-flow-hydrodynamics-17}, instability and the subsequential phase transition in traffic systems have since played a pivotal role across models of different levels of aggregation, specifically, the macroscopic~\cite{traffic-flow-hydrodynamics-14, traffic-flow-hydrodynamics-17, traffic-flow-hydrodynamics-23, traffic-flow-hydrodynamics-25}, mesoscopic~\cite{traffic-flow-btz-01, traffic-flow-btz-03, traffic-flow-btz-lob-02, traffic-flow-btz-lob-04}, and microscopic~\cite{traffic-flow-micro-15, traffic-flow-micro-19, traffic-flow-micro-33, traffic-flow-cellular-04, traffic-flow-micro-37} approaches.

Another crucial feature accompanying the traffic data is its stochasticity, which is attributed to the problem's intrinsic statistical nature.
Deterministic approaches, although intuitive, do not straightforwardly account for such uncertainties or randomness. 
Besides, stochastic modeling is a natural candidate for capturing the uncertainties in the traffic data.
In the literature, the stochastic nature of vehicle traffic has gained widespread recognition in recent years. 
On the microscopic level, for instance, by incorporating the Wiener or Ornstein–Uhlenbeck process into a car-following model, the resulting dynamical system might feature stop-and-go waves, jamiton, or phantom jam as one sits in the linearly unstable region of the parameter space~\cite{traffic-flow-micro-45, traffic-flow-micro-46, traffic-flow-micro-47, traffic-flow-micro-48, traffic-flow-micro-77, traffic-flow-micro-82}.
It is particularly intriguing to point out that in the presence of stochastic noise, the stop-and-go phenomenon can be triggered even when the system is located in a linearly stable region~\cite{traffic-flow-micro-77, traffic-flow-micro-82}.
Similarly, the inclusion of stochastic noise has also been carried out for the cellular automaton~\cite{traffic-flow-cellular-20, traffic-flow-cellular-07, traffic-flow-cellular-25} and hydrodynamic~\cite{traffic-flow-hydrodynamics-52} models.
From the mesoscopic perspective, efforts~\cite{traffic-flow-btz-lob-01, traffic-flow-btz-lob-03, traffic-flow-btz-lob-04} have striven to incorporate this uncertainty into stochastic models.
The primary focus is also on how the stochasticity will dictate the uncertainty and/or modify the stability of the dynamic system, along with other macroscopic features of the traffic flow when compared against their deterministic counterparts. 
One particular feature pertinent to the stochastic nature of traffic systems is the inverse-$\lambda$ shape~\cite{traffic-flow-data-inverse-lambda-01,traffic-flow-data-inverse-lambda-02,traffic-flow-data-inverse-lambda-03,traffic-flow-data-inverse-lambda-04} accompanying capacity drop~\cite{traffic-flow-data-capacity-drop-01,traffic-flow-data-capacity-drop-02,traffic-flow-data-capacity-drop-03} in the fundamental diagram~\cite{traffic-flow-review-02}.
In practice, it typically occurs at active bottlenecks, where the maximum flow is suppressed when queues are formed upstream, reducing bottleneck throughput after congestion onset. 
In the fundamental diagram, an inverse-$\lambda$ shape occurs by a discontinuity in the flow as a function of vehicle concentration, which occurs in the vicinity of the maximum flow.
Consequently, the flow-concentration curve is divided into two different regions of lower and higher vehicle density, respectively, known as {\it free} and {\it congested} flow.
The latter is characterized by significant variance as the data smears over the phase space.
In the literature, the scatter of the data is associated with the stability of the traffic system.
Specifically, it has been either attributed to a transient temporal evolution of a deterministic system such as stop-and-go wave~\cite{traffic-flow-hydrodynamics-23,traffic-flow-hydrodynamics-25,traffic-flow-hydrodynamics-24,traffic-flow-hydrodynamics-22}, or a metastable region in the fundamental diagram such as the so-called {\it synchronized flow} in the context of the three-phase traffic theory~\cite{traffic-flow-three-phase-01,traffic-flow-three-phase-02}.
Nonetheless, debate continues regarding the origin of the scatter~\cite{traffic-flow-data-17,traffic-flow-data-18} and capacity drop~\cite{traffic-flow-hydrodynamics-35,traffic-flow-hydrodynamics-40, traffic-flow-data-capacity-drop-08, traffic-flow-data-capacity-drop-09, traffic-flow-data-capacity-drop-10}.
Via numerical simulation, it was pointed out by some of us~\cite{traffic-flow-btz-lob-04} that, from a mesoscopic viewpoint, the capacity drop might merge as a system's dynamic feature while triggered by varying noise strength.
It was shown that both the scatter and capacity drop are reproduced numerically, which was attributed to the modification to the dynamical model's stability.

Motivated by these considerations, our study aims to scrutinize the topic further with mathematical rigor. 
Using a mesoscopic stochastic model, we analyze the properties of the fundamental diagram, capacity drop, and its inherent uncertainties. 
Specifically, the capacity drop can be interpreted as a trade-off between the system's relaxation towards the deterministic model's stationary states and the stochastic perturbations that undermine this stability. 
Through Itô calculus, we quantitatively analyze the threshold of relevant perturbations that align with this physical depiction. 
By incorporating stochastic noise into a minimalist stochastic fold model, known to reproduce the fundamental diagram's inverse-lambda shape successfully, we demonstrate how stochastic perturbations instigate capacity drops. 
Finally, we explore the system's asymptotic properties and discuss the implications of our findings.

The remainder of the paper is organized as follows.
In the next section, we revisit a two-phase deterministic model based on which the stochastic approach is devised.
In a minimized fashion, the deterministic model incorporates the features of the fundamental diagram consisting of two stable, steady states, namely, the free-flow and congestion states.
In Sec.~\ref{sec3}, the deterministic model will be extended by including stochastic noise in a minimized fashion, giving rise to its stochastic generalizations.
The resulting stochastic mesoscopic model differs from its deterministic counterpart regarding its steady states' stabilities and variance. 
Specifically, we show that the resultant model is physically meaningful regarding the main features of the fundamental diagram and, in particular, the dynamic capacity drop.
The relevant properties and inclusively asymptotic behavior are presented in detail.
We relegate mathematical derivations of the theorems elaborated in the main text to Appx.~\ref{app2}.
In Sec.~\ref{sec4}, numerical simulations are carried out, which manifestly illustrate the physical interpretations of the mathematical theorems.
We explore the model's parameter space, through which a rich structure is revealed.
It is observed that the empirically observed features in the fundamental diagram can be reasonably captured by the model, attributed to an interplay between the few model parameters.
Further analyses, inclusively the numerical validation of the theorems, are relegated to Appx.~\ref{app3} and~\ref{app4}.
The concluding remarks are given in the last section.

\section{A two-phase deterministic model}\label{sec2}

The stochastic mesoscopic traffic model will be established on its deterministic counterpart, a fold catastrophe model first proposed in~\cite{traffic-flow-btz-lob-03}.
In this section, we elaborate on the details of the latter. 
Despite being a bare-bone approach, it features two steady-state solutions, namely, the free-flow and congestion states, both of which are stable.
The equation of motion of the model is presented in Sec.~\ref{sec2.1}, and its solutions are elaborated in Sec.~\ref{sec2.2}.

\subsection{The equations of motion}\label{sec2.1}

Following~\cite{traffic-flow-btz-lob-02}, we consider a deterministic fold catastrophe model that describes a (possibly small) section of the highway with length $L$, which contains $N$ vehicles.
The equation of motion of the system is governed by the rate of change of the occupation number $n_i$ of the {\it{i}}th state with speed $v_i$, which reads
\bqn\label{eqDetEoS}
 \frac{dn_1}{dt}   &=& -c_1n_1+c_2n_1n_2\frac{1}{N_\mathrm{max}-N} \nb\\
 \frac{dn_2}{dt}   &=& c_1n_1-c_2n_1n_2\frac{1}{N_\mathrm{max}-N}  ,
\eqn
where one considers the simplest scenario, which involves only two speed states $i=1,2$ and $c_1$, $c_2$, $v_1$, $v_2$, and $N_\mathrm{max}$ are model parameters.
In principle, the values of the parameters are governed by a generic microscopic theory reflecting the underlying dynamics, e.g., those tailored to the bottleneck section of a highway.
Without loss of generality, we take $c_1,\, c_2>0$, which ensures both a loss and a gaining mechanism.
Also, one assumes $v_1<v_2$ so that the speed $v_2$ corresponds to that of the free-flow state.

In what follows, we briefly discuss the physical contents of the equation of motion.
The rate of change of the occupation number for speed state ``$1$'' is governed by the first line of Eq.~\eqref{eqDetEoS}, consisting of a loss and gaining terms.
The loss term $c_1n_1$ is proportional to the occupation number of the state.
It describes the process when fast drivers leave this state at a rate determined by the transition coefficient $c_1$.
On the other hand, the gain term is proportional to both $n_1$ and $n_2$.
This is because cautious drivers decelerate owing to the occupation of the higher speed state, and meanwhile, the presence of slower vehicles forces the fast ones to slow down to their speed.
The extra constant factor $\frac{1}{N_\mathrm{max}-N}$ takes the total congestion into account, where the parameter $N_\mathrm{max}$ measures the maximum occupation number. 
As elaborated further below, as $N$ approaches $N_\mathrm{max}$ from below, the deceleration effect becomes more significant.

Eq.~\eqref{eqDetEoS} guarantees the total vehicle number is conserved.
In other words, we have 
\bqn
N=n_1+n_2 .\label{Nconservation}
\eqn 
Therefore, the equation of motion effectively possesses only one degree of freedom.

\subsection{Steady solutions and their physical interpretations}\label{sec2.2}

We can simplify Eq.~\eqref{eqDetEoS} using the total vehicle number conservation Eq.~\eqref{Nconservation}, and we have
\begin{equation}\label{eqDetn1}
\frac{dn_1}{dt}=f(n_1)=-c_1n_1+c_2n_1 (N-n_1)\frac{1}{N_\mathrm{max}-N} .
\end{equation}

Given the physical condition $0\le n_1\le N < N_\mathrm{max}$, Eq.~\eqref{eqDetn1} leads to the follow solution:
\begin{equation}\label{eqDetn1sol}
n_1(t)=\frac{\left[(c_1+c_2)N-c_1N_\mathrm{max}\right]\exp\left[N(c_1+c_2)\left(c_1+\frac{t}{N_\mathrm{max}-N}\right)\right]}
{c_2\exp\left[N(c_1+c_2)\left(c_1+\frac{t}{N_\mathrm{max}-N}\right)\right]-\exp\left[c_1\left(c_1+\frac{t}{N_\mathrm{max}-N}\right) N_\mathrm{max}\right]} ,
\end{equation}
when 
\bqn
N > N_c \equiv \frac{c_1}{c_1+c_2}N_\mathrm{max} .\label{defDetBd}
\eqn
It asymptotically approaches a steady solution
\begin{equation}\label{n1AssCon}
n_1^g\equiv n_1(t\to +\infty)=N-\frac{c_1}{c_2}(N_\mathrm{max}-N) .
\end{equation}
The stability of the solution can be found by analyzing the linearized equation of motion for the deviation $\delta n_1 = n_1-n_1^g$, which is found to be
\begin{equation}\label{eqDetDn1}
\frac{d(\delta n_1)}{dt}
=\left[-c_1+c_2 N\frac{1}{N_\mathrm{max}-N}-2c_2n_1^g\frac{1}{N_\mathrm{max}-N}\right]\delta n_1 .
\end{equation}
It is not difficult to show that the quantity in the bracket is negative definite for the region $N > N_c$. 
Therefore, the above steady solution $n_1^g$ is stable.
In other words, Eq.~\eqref{n1AssCon} can be viewed as an attractor of the dynamical system Eq.~\eqref{eqDetEoS} for $N>N_c$, and the fixed point becomes a unphysical ($n_1<0$) repulsor (unstable) when $N < N_c$.

On the other hand, for $N < N_c $, the relevant solution is given by the trivial attractor
\begin{equation}\label{n1AssFree}
n_1^f\equiv 0.
\end{equation}
One can similarly show that the solution is also stable for the relevant region.

The physical interpretation of the above solutions is straightforward by evaluating the corresponding vehicle flows
\begin{equation}\label{qDef}
q = \frac{n_1 v_1+(N-n_1)v_2}{L} .
\end{equation}
One finds, for the flow-concentration curve,
\begin{equation}\label{qEvo}
q(k)=\left\{  
\begin{array}{ll}
    q^f=kv_2 & N\leq N_c, \\
    q^g=q_c+[v_1-\frac{c_1}{c_2}(v_2-v_1)](k-k_c) & N_c<N\leq N_\mathrm{max}\\
\end{array}
\right.
\end{equation}
where $k=N/L$, $k_c=N_c/L$, and $q_c=k_c v_2$.

According to the first line of Eq.~\eqref{qEvo}, the flow arises linearly as a function of concentration.
As indicated by point ``A'' in Fig.~\ref{fig01_fundamental_det}, it reaches its maximum $q=q_c$ at the concentration $k=k_c$ associated with $N=N_c$.
Subsequently, governed by the second line of Eq.~\eqref{qEvo}, the flow decreases monotonically with increasing concentration until the state of total congestion at $N=N_\mathrm{max}$ for which $n_2=0$.
The latter is guaranteed by the factor $\frac{1}{N_\mathrm{max}-N}$ in the equation of motion Eq.~\eqref{eqDetEoS}.
Therefore, $N_c$ marks the bound between free flow and congestion states.
The corresponding fundamental diagram of this simple model is illustrated in Fig.~\ref{fig01_fundamental_det}.
Therefore, the properties of the flow-concentration curve indicate that $n_1^f$ and $n_1^g$ can be readily interpreted as the free-flow and congestion states.
The above notion of the system's stability becomes somewhat different when one includes stochastic noise.

\begin{figure}[!htb]
\begin{minipage}{200pt}
\centerline{\includegraphics*[width=8cm]{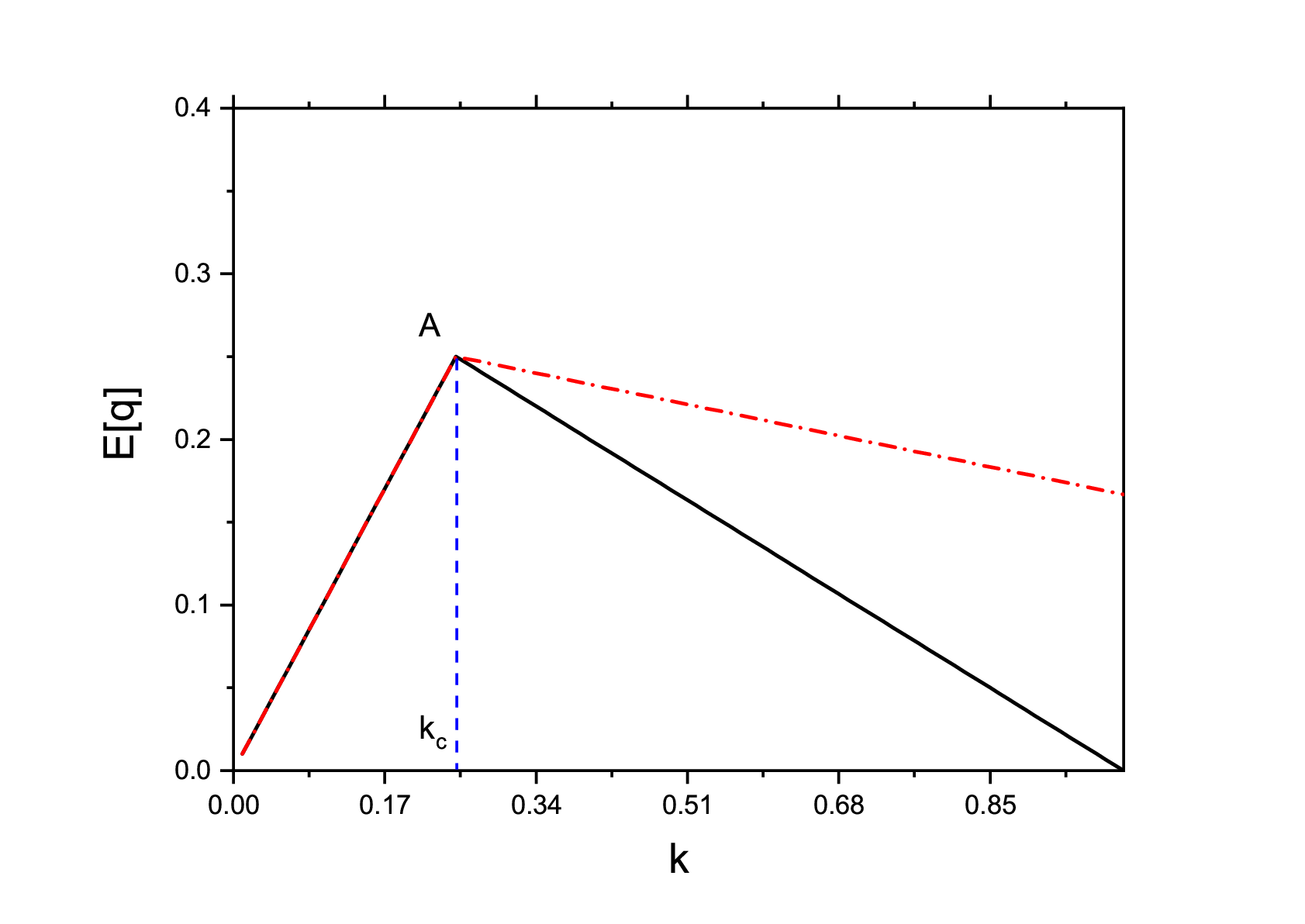}}
\end{minipage}
\caption{The flow-concentration curve for the deterministic two-speed-state model.
As an illustration, the fundamental diagrams are obtained by assuming $v_2=c_1=N_\mathrm{max}=1$, $c_2=3$.
The solid black curve corresponds to the parameter $v_1=0$, while the dashed-dotted red curve corresponds to $v_1=1/6$.}
\label{fig01_fundamental_det}
\end{figure} 

\section{Generalization to a stochastic traffic model}\label{sec3}

The deterministic model presented in the last section can be viewed as a minimized approach that captures the main features of the fundamental diagram.
However, it does not indicate a capacity drop. 
Specifically, the free-flow state immediately becomes unstable once it enters the congestion region, and the connection between the two states is continuous.
Moreover, the model does not feature any uncertainty, and subsequently, the variance of any quantity vanishes.
Regarding a more realistic traffic flow scenario, its stochastic nature and capacity drop constitute its crucial characteristics.
The main goal of the remainder of the present study is to demonstrate that the above features can be realized by introducing stochastic noise and, again, via a minimized fashion.

Now, let us introduce stochastic noise on top of the above deterministic version of the model.
Analytic properties of the resultant dynamic system will be explored, for which our analysis will mainly reside on the analytic side.
The mathematical aspect of the stochastic differential equation will be presented in terms of a few theorems, which are further elaborated in Appx.~\ref{app2}.
More importantly, we elaborate on the insights from a physical perspective, and particular attention is given to the modification to the system's stability.
The latter is mainly non-trivial and physically intriguing to effectively capture the realistic features observed in traffic data.

\subsection{Stochastic differential equation}\label{sec3.1}

As discussed earlier, we aim to introduce stochastic dynamics in a minimized fashion.
To this end, we introduce stochastic noise to the gain term in the first line of Eq.~\eqref{eqDetEoS} while ensuring the total vehicle number conservation.
Specifically, one replaces 
\bqn
c_2 \to c_2+\sigma \dot{B}(t) ,\label{EqStoAsum}
\eqn
where $\sigma$ is the intensity noise and $B(t)$ is a standard Brownian motion independent of the initial condition.
This corresponds to certain uncertainties when a driver speeds up their vehicle while the remaining transitions are deterministic.

Subsequently, the dynamics of the system are governed by the following stochastic differential equation
\begin{equation}\label{SDE}
    dn_1(t)=n_1(t)\left[ \left(-c_1+c_2(N-n_1(t))\frac{1}{N_\mathrm{max}-N}\right)dt+\sigma (N-n_1(t))\frac{1}{N_\mathrm{max}-N}dB(t) \right] ,
\end{equation}

From a mathematical perspective, the asymptotic properties of the system can be captured by a few theorems, which can be derived by partly following the strategy developed in~\cite{khasminskii, Gray}.
In the remainder of this section, we will first present the relevant theorems and then elaborate on their physical insights into the traffic system. 
We relegate further details of the derivations of the theorems to Appx.~\ref{app2}.

\subsection{The stability of the free-flow state}\label{sec3.2}

Now, Eq.~\eqref{SDE} is a stochastic differential equation.
The stability of the initially steady states of its deterministic counterpart Eq.~\eqref{eqDetn1} must be understood regarding the notion of convergence of random variables.
For our present purpose, we do not intend to delve into the relatively subtle difference between convergence in probability and that in distribution~\cite{book-statistical-inference-Wasserman} when applying the following theorems to realistic scenarios. 
Specifically, we will, by and large, be content with ourselves if convergence in distribution is achieved. 

First, the following lemma guarantees that the solution will be physical for the relevant time domain.
In other words, even though persistent, stochastic noise will not perturb the system into an unphysical region.
\begin{Lem}\label{invariance}
    For any given value $n_1(0)\in(0,N)$, the stochastic differential equation (\ref{SDE}) has a unique global positive solution such that:
\begin{equation}
    \mathbb{P}(n_1(t)\in (0,N),\forall t\geq 0)=1 .
\end{equation}
\end{Lem}

\begin{definition}
  Let $R_0^s:=\frac{\alpha c_2N}{c_1}- \frac{\alpha^2\sigma^2N^2}{2c_1}$, where $\alpha=\frac{1}{N_\mathrm{max}-N}$. 
\end{definition}

As it turns out, the quantity $R_0^s$ plays a crucial role in characterizing different traffic states.
In particular, one might be interested in the emergence of an asymptotic free-flow state at the limit $t\to\infty$.
Using the above definition, the impact on the stability of the free-flow state Eq.~\eqref{n1AssFree} can be inferred from the following theorem.

\begin{Thm}\label{2a3}
If $R_0^s<1$ and $\sigma^2< \frac{c_2}{\alpha N}$, then for any initial condition $n_1(0)\in(0,N)$, the solution of (\ref{SDE}) satisfies

\begin{equation}\label{eqFreeFlow1}
    \limsup_{t\to\infty}\frac{1}{t}\log(n_1(t))\leq \alpha c_2N-c_1-\frac{1}{2}\alpha^2\sigma^2N^2<0 \hspace{.2cm}\text{a.s.}
\end{equation}
where ``a.s.'' stands for almost sure convergence.
\end{Thm}

The above theorem furnishes a sufficient condition for the system to evolve into a free-flow state as the system stabilizes in time.
In what follows, we elaborate further on the physical implications.

First, theorem~\ref{2a3} guarantees the presence of a stable free flow state as $n_1$ asymptotically approaches $0$ faster than some exponential function as $t\to \infty$\footnote{At a first glimpse, one might think $n_1<1$ suffices for Eq.~\eqref{eqFreeFlow1}.
However, Eq.~\eqref{eqFreeFlow1} implies an even stronger statement as it implies the existence of some $\epsilon>0$ that, intuitively, $n_1(t)<\exp(-\epsilon t)$ for sufficiently large $t$.}.

Such free-flow states occur at insignificant vehicle occupation and stochastic noise, which are physically pertinent.
In fact, given these conditions, one can proceed further with some quantitative estimations.
The limit of small vehicle occupation can be interpreted mathematically as $N\ll N_\mathrm{max}$.
The condition $R_0^s <1$ roughly implies that $N< \frac{c_1}{c_1+c_2}N_\mathrm{max}$ since the first term of $R_0^s$ is dominant when the stochastic noise is negligible.
Specifically, the latter assumption indicates $\sigma^2 < \frac{N_\mathrm{max}-N}{N}c_2$.
Therefore, using the above two inequalities, straightforward algebra gives us a rough estimation of a sufficient bound of the stable free-flow state.
To be more specific, it is governed by the lesser of two values 
\begin{equation}\label{bdFreeFlow}
N_\mathrm{bound}= \mathrm{min}\left(N_c, N_s\right) ,
\end{equation}
where $N_c$ is defined in Eq.\eqref{defDetBd}, and
\begin{equation}\label{NsDef}
N_s\equiv\frac{c_2 N_\mathrm{max}}{\sigma^2+c_2} .
\end{equation}
On the one hand, one notices that $N_c$ is nothing but the bound of the deterministic model defined by Eq.~\eqref{defDetBd} in Sec.~\ref{sec2.2}.
On the other hand, the stochastic one $N_s$ is found to decrease as the strength of the noise strength $\sigma$ increases, which is also physically sound.
The above dependence on the strength of stochastic noise will only become relevant when $\sigma^2 > \frac{c_2^2}{c_1}$.
In other words, stochastic noise must be strong enough until it plays a role in undermining the free-flow state's stability.

Moreover, it is not difficult to show that the bound estimated above can be further improved by explicitly considering the correction due to the second term of $R_0^s$.
Formally, a refined version of Eq.~\eqref{bdFreeFlow} reads
\begin{equation}\label{bdFreeFlowRef}
N_\mathrm{bound}= \mathrm{min}\left(N_c', N_s\right) ,
\end{equation}
where the deterministic bound $N_c$ suffers a correction due to the stochastic noise, namely,
\begin{equation}\label{Ncprime}
N_c' = N_c + \Delta N_c .
\end{equation}
The term $\Delta N_c$ is owing to the correction of the stochastic noise, which is positive definite
\begin{equation}\label{DeltaNc}
\Delta N_c > 0 .
\end{equation}
Although well defined, the specific form of $\Delta N_c$ is tedious.
However, a reasonably good approximate form can be obtained when $\sigma^2$ is small.
Specifically, we have
\begin{equation}\label{DeltaNcEst}
\Delta N_c \sim \frac{c_1\sigma^2}{2c_2(c_1+c_2)} .
\end{equation}

In the context of the present model, the mathematical fact that we always have the relation $N_c'> N_c$ in the presence of moderate stochastic noise can be readily interpreted as {\it capacity drop}, an intriguing feature from a physical perspective.
In other words, such a relation implies that the free-flow state stretches beyond its deterministic counterpart and above the congestion state.
The above result is counterintuitive.
This is because it seems more plausible that external perturbations would destabilize a state.
Conversely, Eqs.~\eqref{Ncprime} and~\eqref{DeltaNc} indicate that the free-flow state's stability is further strengthened due to stochastic noise.
Nonetheless, as Eq.~\eqref{NsDef} dictates, $N_s$ still destabilizes the free-flow state.
The eventual outcome is a balance between the two competing effects.

Before closing this subsection, it is also instructive to consider the other extreme, namely, the region when $N\to N_\mathrm{max}$.
In particular, one might worry that the second term $\frac{\alpha^2\sigma^2N^2}{2c_1}$ overwhelms the first term $\frac{\alpha c_2 N}{c_1}$ in $R_0^s$, so that $R_0^s<0$.
The latter may undermine the model's validity as it indicates the emergence of stable free flow at unreasonably high vehicle occupation, a physically unrealistic scenario.
Fortunately, it can be shown that such a scenario will not occur.
To see this, one equates the two terms to give $\sigma^2=\frac{2c_2}{\alpha N}$.
However, since free flow demands $\sigma^2 < \frac{c_2}{\alpha N}$, the latter condition eliminates the possibility of obtaining free flow at a significant vehicle density.
Since both terms diverge as $N\to N_\mathrm{max}$, substituting the limiting value $\sigma^2 = \frac{c_2}{\alpha N}$ gives $R_0^s \gg 1$.
In other words, theorem~\ref{2a3} correctly associates the free-flow state with low vehicle occupations, as one would have expected.

\subsection{Flux persistence in the congestion state}\label{sec3.3}

The theorems and lemma given in the following two subsections are regarding the properties of the congestion traffic state.
The present subsection is devoted to the persistence and randomness in the congestion state, while the following subsection elaborates on the stationary aspect of the congestion flow.

\begin{Thm}\label{flux1}
If $R_0^s>1$, then for any given initial value $n_1(0)\in(0,N)$, the solution of (\ref{SDE}) satisfies
\begin{equation}\label{first}
       \limsup_{t\to\infty}n_1(t)\geq \xi\hspace{.2cm}a.s.,
\end{equation}
and
\begin{equation}\label{second}
        \liminf_{t\to\infty}n_1(t)\leq \xi\hspace{.2cm}a.s.,
\end{equation}
where 
\begin{equation}\label{xi}
    \xi=\frac{1}{\alpha^2\sigma^2}\left[ \sqrt{\alpha^2c_2^2-2\alpha^2\sigma^2c_1}-(\alpha c_2-\alpha^2\sigma^2N)\right].
\end{equation}
\end{Thm}

The above theorem concerning limits inferior and superior indicates that the congestion state typically features a finite variance, namely, uncertainty in vehicle occupation.
Moreover, as discussed further below and in the appendix, such a state is intrinsically neither static nor stationary but can be asymptotically characterized by a probability distribution.
Theorem~\ref{flux1} implies that, as the system evolves, it will repeatedly traverse a specific vehicle occupation, denoted by $\xi$, infinitely often with probability one.
The following corollary examines the influence of the stochastic noise’s strength $\sigma$.

\begin{Cor}\label{lemTen}
Suppose that $R_0^s>1$, and consider $\xi$ defined by Eq.~\eqref{xi} as a function of $\sigma$.
For the interval
 \begin{equation}
 0<\sigma<\tilde{\sigma}  , 
 \end{equation} 
$\xi$ is strictly decreasing and attains the limits
\begin{equation}
 \lim_{\sigma\to 0^+}\xi=N\left( 1-\frac{c_1}{\alpha c_2 N} \right) ,\label{limZeroSig}
\end{equation} 
and
\begin{equation}
 \lim_{\sigma\to\tilde{\sigma}^{-}}\xi=\left\{
 \begin{array}{ll}
   0   & \text{if $1\leq R_0 \leq 2$}, \\
    N\left(\frac{R_0-2}{R_0-1}\right)  & \text{if $R_0>2$ },\\
 \end{array}
 \right.
\end{equation}
where $\tilde{\sigma} = \frac{\sqrt{2(\alpha c_2N-c_1)}}{N}$ and $R_0=\frac{\alpha c_2 N}{c_1}$.
\end{Cor}

Some comments are in order.
Theorem~\ref{flux1} implies that the vehicle concentration fluctuates with a superior bound above $\xi$ while its inferior bound is below $\xi$.
Based on the above discussions, $R_0^s>1$ largely corresponds to the congestion state.
Let us first consider the scenario with vanishing stochastic noise $\sigma \to 0^+$.
At this limit, Eq.~\eqref{xi} gives rise to Eq.~\eqref{limZeroSig}, namely, $\xi=\frac{N(c_1+c_2)-N_\mathrm{max}c_1}{c_2}$.
In particular, it attains $\xi=0$ at $N=N_c$ (where $R_0^s=1$) and increases linearly with $N$ for $N>N_c$ until $\xi=N_\mathrm{max}$ at $N=N_\mathrm{max}$.
This agrees with the intuition that the system falls back to the flow-concentration diagram demonstrated in Fig.~\ref{fig01_fundamental_det} under the above circumstances.

Now, for $0<\sigma^2 < \tilde{\sigma}^2$ when stochastic noise starts to play a role, corollary~\ref{lemTen} dictates that $\xi$ mainly sits beyond the above line segment.
However, it is readily verified that one still has $\xi=0$ at $N=N_c$ and $\xi=N_\mathrm{max}$ at $N=N_\mathrm{max}$.
Therefore, it seems inviting to speculate that the flow capacity, by and large, forms a convex shape when compared to the picture of a deterministic dynamic system.

Moreover, intuitively, owing to the stochastic uncertainty, the data are expected to spread about the deterministic model.
A quantitative estimation of the flow variance will be given below by corollary~\ref{lemTen} and theorem~\ref{meanandvariance}.
Further insights will be elaborated by applying the above theorem in conjunction with those in the following subsection.
Before closing this subsection, we emphasize the interpretation that the congestion-flow state is not constituted by steady-state solutions.
This is because, for such a system, it evolves $n_1(t)$ perpetually as it traverses the specific value $\xi$ infinitely, often with probability one.
From an observational perspective, the vehicle's concentration will constantly traverse a given flow state denoted by $n_1=\xi$.
To appropriately describe such a traffic state, one employs the notion of statistical distribution, as discussed below.

\subsection{Stationary distribution in the congestion state}\label{sec3.4}

Let $\mathbb{P}_{0}(t)$ be the probability measured induced by $n_1(t)$ with initial value $n_1(0)$, more specifically 

\begin{equation}
\mathbb{P}_0(t)(A)=\mathbb{P}(n_1(t)\in A|n_1(0)),
\end{equation}
for all $A\in \mathcal{F}$, where $\mathcal{F}$ is the $\sigma$-field of Borel~\cite{khasminskii} on $\Omega=(0,N)$. 
If there is a probability measure $\pi(\cdot)$ on the measurable space $((0, N), \mathcal{F})$  such that $\mathbb{P}_0(t)$ converges weakly to $\pi$ as $t\to\infty$, then $\pi$ will be referred to, below in theorems~\ref{stationary} and~\ref{meanandvariance}, as a stationary distribution or invariant measure of the dynamical system. 

In practice, assigning probabilities to all subsets of the sample space $\Omega$ is often not feasible.
This is why one resorts to dealing with a specific class of sets, namely, the Borel $\ sigma$ field.
In this subsection, we discuss the feasibility of the stationary distribution in the congestion traffic state in our model, elaborated by the following theorems.

\begin{Thm}\label{stationary}
  If $R_0^s>1$, the dynamic system governed by Eq.~\eqref{SDE} has a unique stationary distribution.
\end{Thm}

\begin{Thm}\label{meanandvariance}
Assume that $R_0^s>1$ and let $\mu$ and $\gamma$ denote the mean and variance of the stationary distribution of the system Eq.~\eqref{SDE}, we have
\begin{equation}\label{mean}
\mu=\frac{2 c_2(R_0^s-1)c_1}{2 c_2(\alpha c_2-\alpha^2\sigma^2N)+\alpha\sigma^2(\alpha c_2N-c_1)}
\end{equation}
and
\begin{equation}\label{variance}
\gamma=\frac{\mu(\alpha c_2N-c_1)}{\alpha c_2}-\mu^2.
\end{equation}
\end{Thm}

As the system evolves, it does not asymptotically approach any given state.
Nonetheless, its probability distribution asymptotically approaches a well-defined distribution that does not depend on time.
Intuitively, randomness is expected from a stochastic dynamic system.
The above theorems further show that such randomness, by and large, can be quantitatively captured.

Again, to gain some physical insights, let us explore the limits of the above results.
If we assume the limit $\sigma^2 \to 0$, 
\begin{equation}\label{limmuN}
\lim\limits_{\sigma^2\to 0}\mu = N-\frac{c_1}{c_2}(N_\mathrm{max}-N) = n_1^{g} ,
\end{equation}
where the r.h.s. is compared against Eq.~\eqref{n1AssCon}, and
\begin{equation}\label{limgammaN}
\lim\limits_{\sigma^2\to 0}\gamma = 0 .
\end{equation}

In particular, at the peak of the deterministic flow-concentration curve, one finds
\begin{equation}
\lim\limits_{N\to N_c}\left(\lim\limits_{\sigma^2\to 0}\mu\right) = 0 ,\label{muLim1}
\end{equation}
and
\begin{equation}
\lim\limits_{N\to N_c}\left(\lim\limits_{\sigma^2\to 0}\gamma\right) = 0 .\label{gammaLim1}
\end{equation}
While at the maximal concentration $N\to N_\mathrm{max}$, we have 
\begin{equation}
\lim\limits_{N\to N_\mathrm{max}}\left(\lim\limits_{\sigma^2\to 0}\mu\right) = N_\mathrm{max} ,\label{limmu0}
\end{equation}
and
\begin{equation}
\lim\limits_{N\to N_\mathrm{max}}\gamma = 0 ,\label{limgamma0}
\end{equation}
where one notes that the last expression does not demand the limit $\sigma^2\to 0$.
The above results by Eqs.~\eqref{limmuN}-\eqref{limgamma0} are somewhat expected as the limits fall back to their deterministic counterparts.

Of course, we are more interested in scenarios with nonvanishing stochastic noise.
If we do not take the limit $\sigma^2\to 0$ but still assume $\sigma^2 \ll 1$, it is not difficult to notice that the condition $R_0^s>1$ gives a bound slightly above the location of the deterministic peak, namely, $N\gtrsim N_c$.
In this case, in addition to Eq.~\eqref{limgamma0}, we have
\begin{equation}
\lim\limits_{N\to N_c}\mu = -\frac{c_1^2 N_\mathrm{max}\sigma^2}{2(c_1+c_2)(c_2-c_1\sigma^2)} ,\label{muLim2}
\end{equation}
\begin{equation}
\lim\limits_{N\to N_\mathrm{max}}\mu = N_\mathrm{max} ,\label{limmu}
\end{equation}
and
\begin{equation}
\lim\limits_{N\to N_c}\gamma = -\frac{c_1^4 N_\mathrm{max}^2\sigma^4}{4(c_1+c_2)^2(c_2-c_1\sigma^2)^2} .\label{limgamma2}
\end{equation}
Although tedious, the derivation of the above results is relatively straightforward.
We also note that the condition $R_0^s >1$ is marginally violated at $N=N_c$, which is primarily responsible for the {\it incorrect} signs found on the r.h.s. of Eqs.~\eqref{muLim2} and~\eqref{limgamma2}.
Nonetheless, as an educated estimation, we conclude that the values for $\mu$ and $\gamma$ are significantly suppressed at the emergence of the congestion state, which will be reinforced by corollary~\ref{lemMu} below.

The physical interpretations implied by Eqs.~\eqref{limgamma0}-\eqref{limgamma2} are desirable.
They indicate that the stochastic dynamic system naturally satisfies the relevant physical constraints.
On the one hand, at the emergence of congestion flow, both the flow and its variance asymptotically fall back to those of the deterministic model.
On the other hand, at total congestion, the system again approaches its deterministic counterpart as it falls to the {\it ground state}, where all vehicles move at the lowest speed possible.
Subsequently, the stochastic system asymptotically approaches a deterministic state where the statistical fluctuations vanish.
In other words, Eqs.~\eqref{limgamma0}-\eqref{limgamma2} indicate that the flow approaches its deterministic counterpart at both ends of the congestion state, and meanwhile, its variance forms an almond shape that vanishes at both ends of the congestion state.

Last, one can extract further information on the constraint of the mean and variance by the following corollary.
\begin{Cor}\label{lemMu}
Consider $\mu$ defined by Eq.~\eqref{meanandvariance} as a function of $\sigma$,
$\mu$ is strictly decreasing with increasing $\sigma$, and moreover
\bqn
0\le \mu \le n_1^g ,\label{muBound}
\eqn
and
\bqn
0\le \gamma \le \frac14 \left(n_1^g\right)^2 .\label{gammaBound}
\eqn
The maximum of $\gamma$ is attained at $\mu=\frac12 n_1^g$.
\end{Cor}
Since $n_1^g=0$ at $N=N_c$, apart from its finite variance, the average flow coincides with the deterministic model at $N=N_c$.
Since flow increases with decreasing $n_1$, as $N$ increases, the average flow decreases but always stays above its deterministic counterpart owing to Eq.~\eqref{muBound}.
However, due to Eq.~\eqref{limmu}, the flow again coincides with the deterministic model at the maximal concentration $N=N_\mathrm{max}$.
This reiterates the results drawn from theorem~\ref{meanandvariance} and is consistent with the convex shape of the fundamental diagram discussed in the previous subsection.

\section{Numerical simulations and further discussions}\label{sec4}

This section explores the mathematical theorems elaborated in the previous section from a numerical perspective.
Through simulations using the \textit{Python} package \textit{sdeint} to numerically solve the stochastic differential equation~(\ref{SDE}), we assess their feasibility and observable implications to the stochastic two-speed-state traffic model.

First, we numerically verify the lemmas, theorems, and corollaries given in the main text.
For conciseness, we relegate these analyses to Appx.~\ref{app3}.
We then explore the parameter space of the proposed model while reflecting on the relevant features of the observed traffic flow elaborated in the previous section.
Specifically, even as a minimized approach, we show that it is feasible to reproduce the main features observed in the traffic data by appropriately choosing the model parameters.

Before presenting the numerical results, we elaborate on the relevant ranges of the model parameters that will be explored for the numerical simulations.
For the free-flow state elaborated at low vehicle concentration elaborated by theorem~\ref{2a3} to be governed by the deterministic bound $N_c$, one requires
\bqn
\sigma^2 \le \frac{c_2^2}{c_1} ,\label{SigCond1}
\eqn
namely, $N_c<N_s$, according to Eq.~\eqref{bdFreeFlow}.
Moreover, for the estimation Eq.~\eqref{DeltaNcEst} to be valid, one needs
\bqn
\sigma^2 \ll 1.
\eqn
Since $N_\mathrm{max} \gg 1$ and $\frac{c_2}{c_1}\sim O(1)$, we have roughly
\bqn
\sigma^2 \ll \frac{c_2^2}{c_1}N_\mathrm{max} \sim c_2 N_\mathrm{max} .
\eqn

At large vehicle concentration, it turns out that the model might be plagued by nonphysical solutions owing to the following theorem:
\begin{Thm}\label{2a4}
  If $\sigma^2>\frac{ c_2}{\alpha N}\vee \frac{c_2^2}{2c_1}$, the solution of (\ref{SDE}) satisfies
  \begin{equation}
      \limsup_{t\to\infty}\frac{1}{t}\log(n_1(t))\leq -c_1+\frac{c_2^2}{2\sigma^2}<0\hspace{.2cm}\text{a.s.}
  \end{equation}
\end{Thm}
The above theorem indicates a nonphysical free-flow state, occurring largely at large vehicle concentrations. 
To avoid such a scenario, one can introduce a truncation to the vehicle concentration $N_\mathrm{cut}$ that satisfies the following condition
\bqn
N < N_\mathrm{cut}\lessapprox N_\mathrm{max} ,
\eqn
so that
\bqn
\sigma^2 < \frac{c_2(N_\mathrm{max}-N_\mathrm{cut})}{N_\mathrm{cut}} ,\label{SigCond2}
\eqn
or effectively, $N_\mathrm{cut} < N_s$.
The constraint Eq.~\eqref{SigCond2} should be used in conjunction with Eq.~\eqref{SigCond1}.

While bearing in mind the above model constraints, we investigate the fundamental diagram's parameter dependence.
In particular, we explore the role of the parameters $c_1$, $c_2$, and $\sigma$ in the flow-concentration curve.
The results are shown in Fig.~\ref{fig:figure18}.
The temporal evolution of the system is governed by the stochastic differential equation~\eqref{SDE}.
The initial condition $n_1(0)$ is drawn from a uniform distribution $\mathrm{U}(1, N)$.
A total of 3,000 simulations have been carried out for each plot, where 20 simulations are performed for a given set of parameters.
For each simulation, the time instant $t_s\sim 26$ is chosen to evaluate the flow and concentration, which is drawn randomly from the uniform distribution $\mathrm{U}(25, 27)$.
Each gray dot represents the resulting flow of a given simulation, and the dashed red curve indicates their average. 
In contrast, the blue solid curve gives the underlying deterministic model.
For a more transparent comparison between the plots, we assume a given value for the deterministic bound $N_c= 50$ and choose $N_\mathrm{max}=200$ and $N_\mathrm{cut}=150$. 
The reason for choosing $N_c$ instead of any specific values for $c_1$ and $c_2$ is that the overall shape of the fundamental diagram, and inclusively the location of its maximum, is primarily governed by the observation.
We also note that the remaining parameters are determined up to a relative scaling.
Subsequently, to generate the plots shown in Fig.~\ref{fig:figure18}, one varies the parameters $c_1$ and $\sigma$ in the vertical and horizontal directions, and subsequently, $c_2$ is governed by Eq.~\eqref{defDetBd}.
We relegate more detailed information on the numerical simulations to Appx.~\ref{app3}, and the impact of the chosen parameters is further scrutinized in Appx.~\ref{app4}.

\begin{figure}[!ht]
\begin{minipage}{440pt}
\centerline{\includegraphics[width=1.0\textwidth]{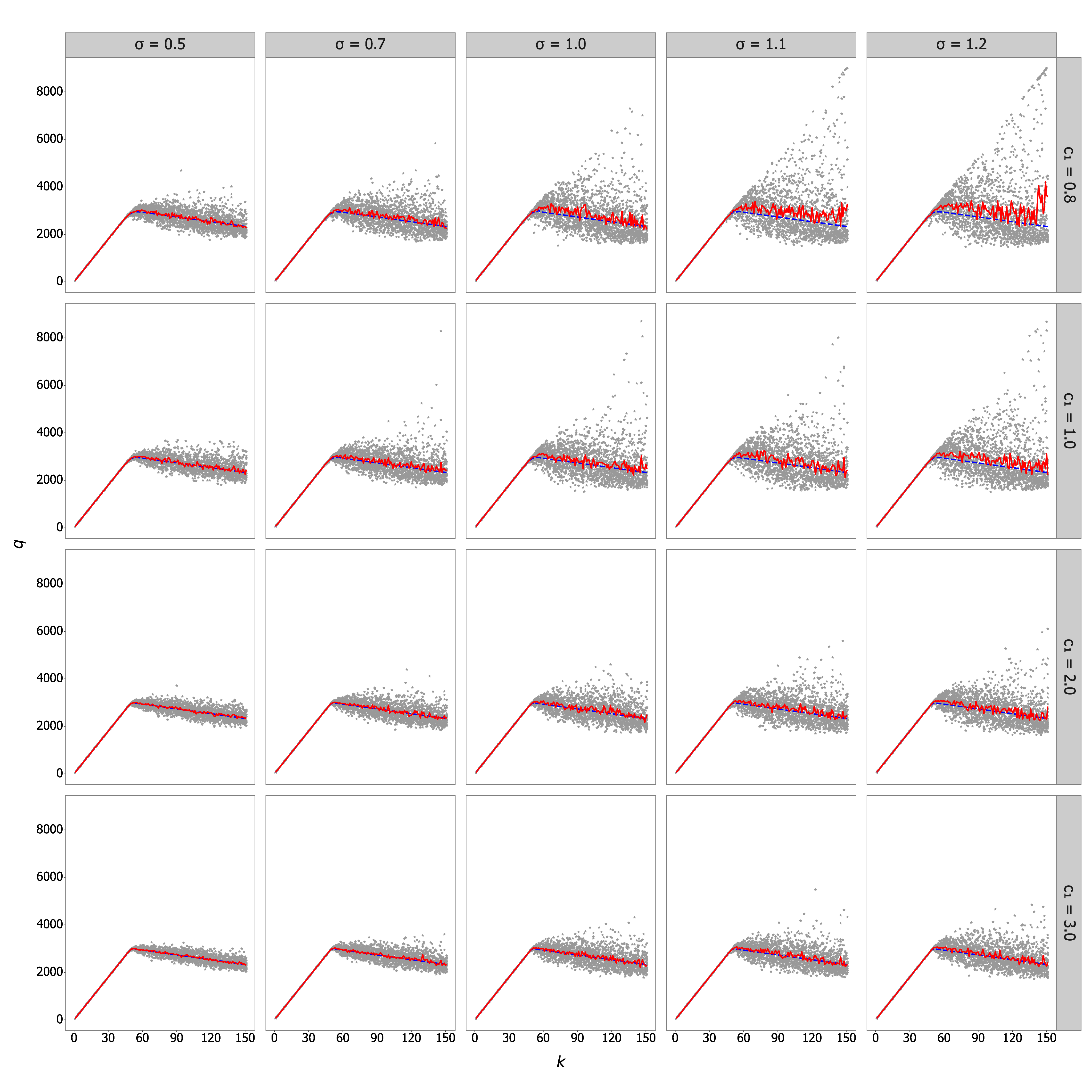}}
\end{minipage}
\renewcommand{\figurename}{Fig.}
\caption{The flow-concentration curve for the stochastic two-speed-state model.
A comprehensive scan of the model parameter space is conducted.
For a given setup, each gray dot represents the resulting flow of a given simulation, and the dashed red curve indicates their average.
The blue solid curve shows the fundamental diagram of the underlying deterministic model.
In the calculations, we vary $\sigma$ and $c_1$ while fixing the remaining parameters $v_1=10$, $v_2=60$, $L=1$, $N_\mathrm{max}=200$, and $N_c = 50$.
For illustration purpose, we choose a constant value $N_\mathrm{cut}=150$, where the inequality Eq.~\eqref{SigCond2} is not always guaranteed.} 
\label{fig:figure18}
\end{figure}

Let us first explore the effect of the parameters regarding the aleatory uncertainty in the flow-concentration curve and, in particular, the capacity drop.
It is pretty intuitive that increasing $\sigma$ results in more significant uncertainties.
While the effect of $c_1$ on the intrinsic variation or randomness of the fundamental diagram is more subtle.
The effect of capacity drop becomes less pronounced as $c_1$ increases while keeping $N_c$ unchanged.
This might be less obvious at first glance, but it is readily shown that such a feature agrees with Eq.~\eqref{DeltaNcEst}.
Second, these parameters also affect the bound of the truncation $N_\mathrm{cut}$ applied to the model.
While the remaining parameters are kept unchanged, as $N_\mathrm{cut}$ increases, the value of $\sigma$ satisfying Eq.~\eqref{SigCond2} will decrease.
Using similar arguments, Eq.~\eqref{SigCond2} implies that the lower bound of $N_\mathrm{cut}$ increases with increasing $c_1$, for given $\sigma$ and $N_\mathrm{max}$.
This explains why the top-right plots of Fig.~\ref{fig:figure18} display nonphysical free-flow states at more significant vehicle concentration.

As a result, there is competition between the two mechanisms.
On the one hand, the value of $ c_1 $ cannot be too small, as it might undermine the model's validity owing to theorem~\ref{2a3}.
On the other hand, to have a more pronounced capacity drop in accordance with the data, it is desirable to choose a moderate value for $c_1$ that is not too large.
A similar argument applies to the choice of $\sigma$.
The significant scatter of the fundamental diagram observed in the congestion phase indicates that a more significant strength of the stochastic noise is favorable, while the latter also possibly invalidates the model.
The scan of model parameters presented in Fig.~\ref{fig:figure18} indicates that a range of model parameters furnishes a physically relevant traffic model, which will be further elaborated below.

\begin{figure}[ht]
\begin{tabular}{ccc}
\begin{minipage}{170pt}
\centerline{\includegraphics[width=1\textwidth]{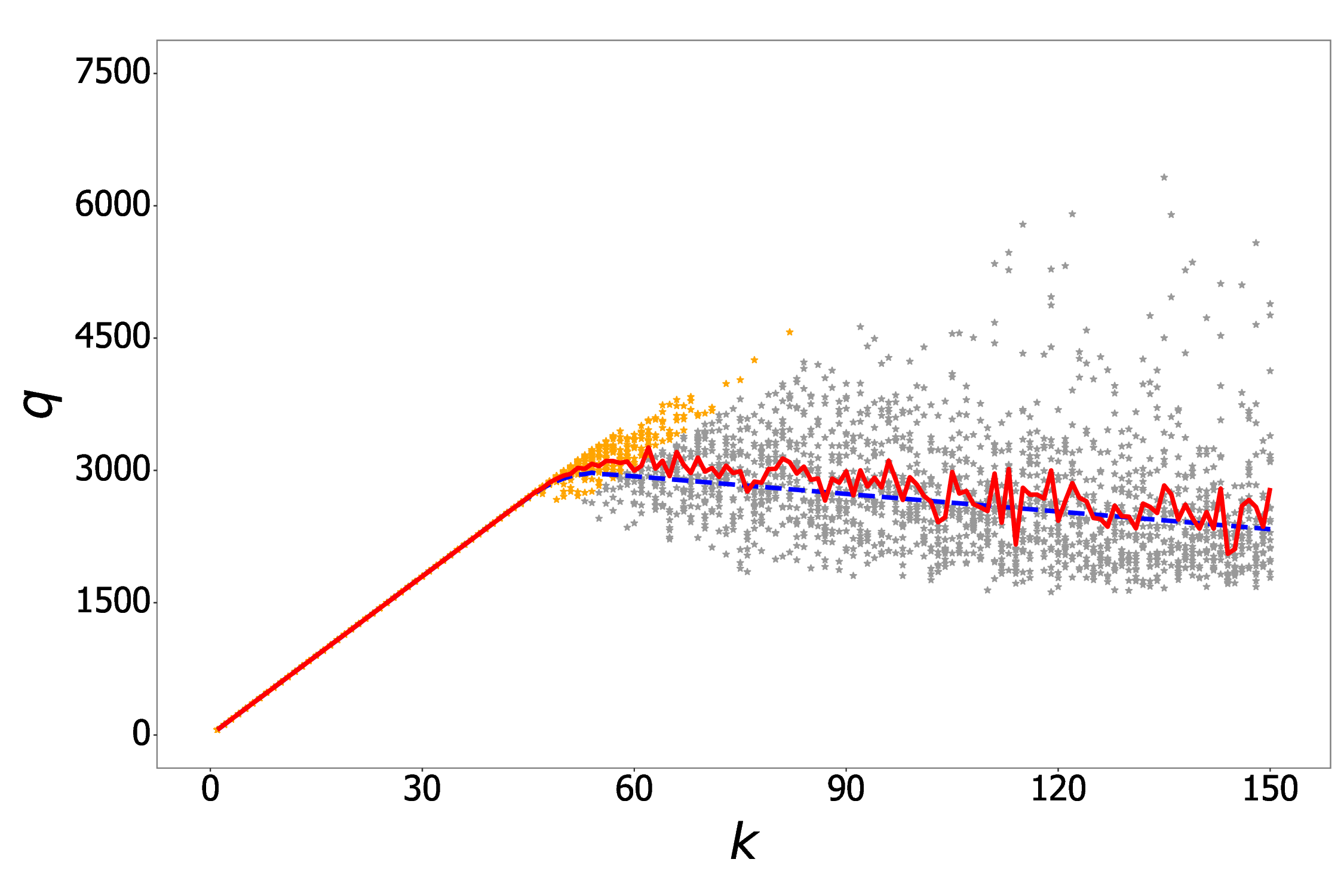}}
\end{minipage}
&
\begin{minipage}{170pt}
\centerline{\includegraphics[width=1\textwidth]{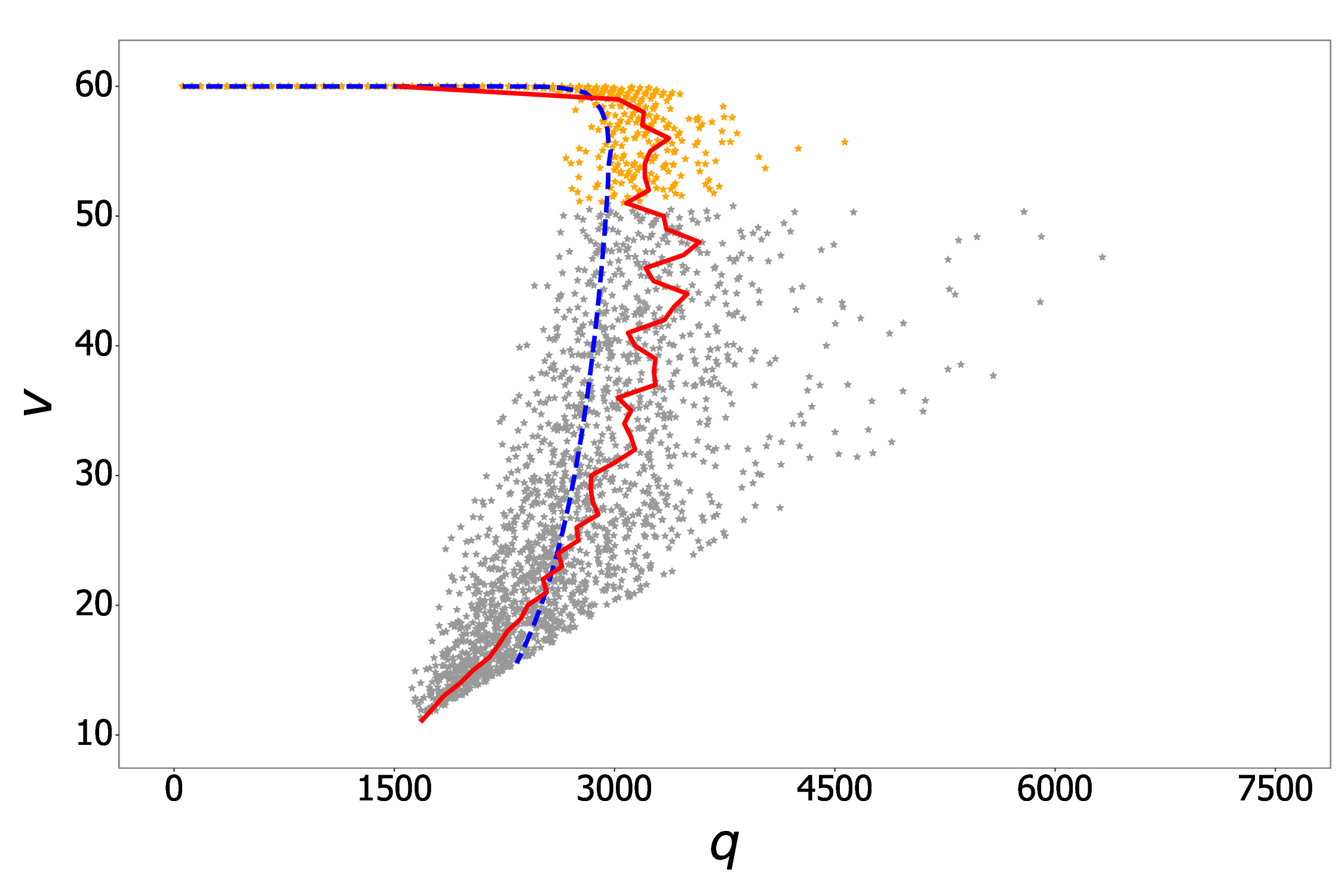}}
\end{minipage}
&
\begin{minipage}{170pt}
\centerline{\includegraphics[width=1\textwidth]{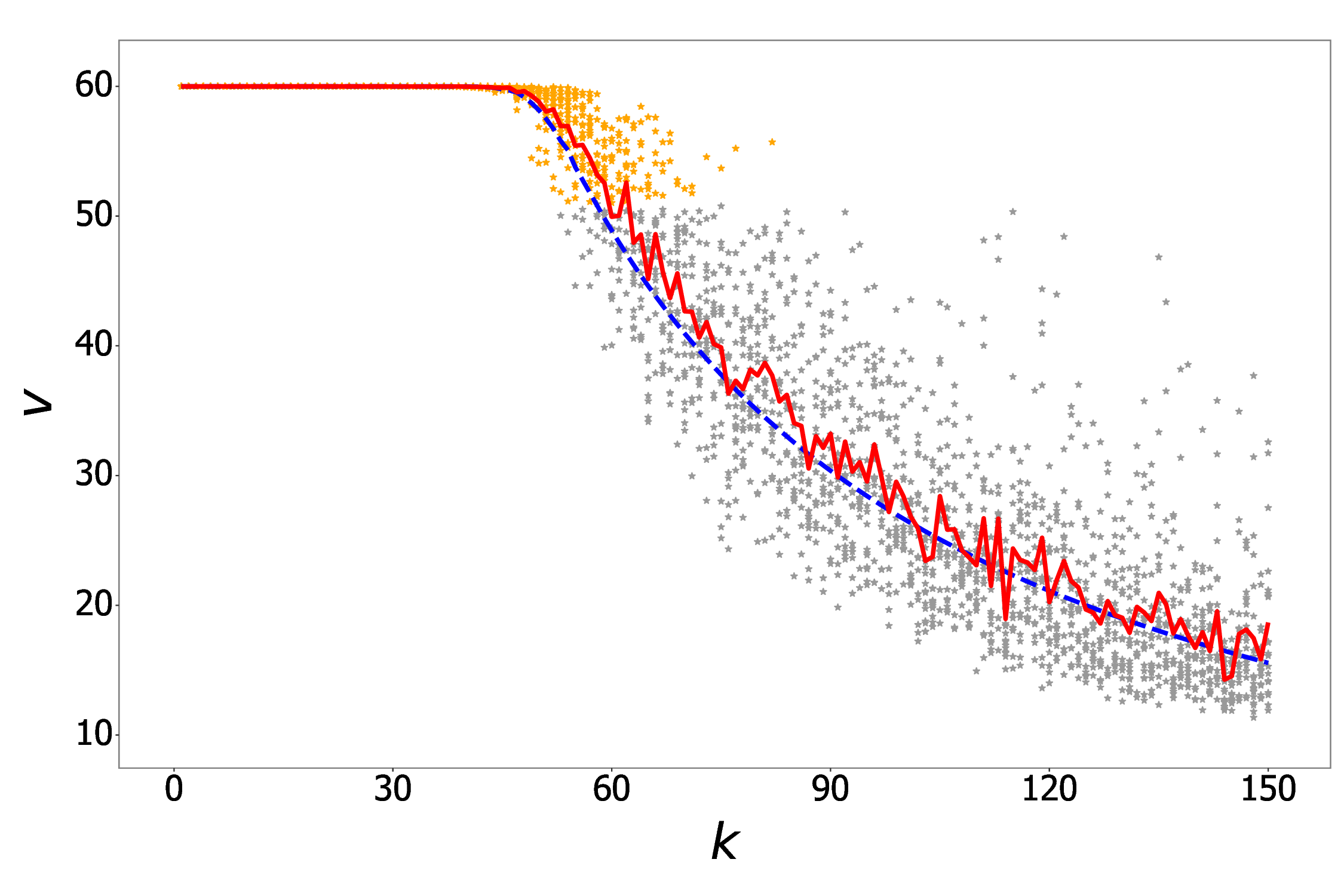}}
\end{minipage}
\end{tabular}
\renewcommand{\figurename}{Fig.}
\caption{The results of the stochastic two-speed-state model, where the calculations are carried out using the parameters Eq.~\eqref{SModelParameters}.
Left: the flow-concentration curve.
Middle: the traffic flow as a function of average vehicle speed.
Right: the vehicle concentration as a function of average vehicle speed.
In the plots, each gray or orange symbol represents the result of a given simulation.
The orange squares represent cases where the flow reaches at least 85\% of the deterministic model's free flow (i.e., $q \geq 0.85 q^f = kv_2$), primarily indicating free-flow conditions, while gray dots denote the remainder.
The dashed red curve indicates the average value for the stochastic model, while the solid blue curve shows that of its deterministic counterpart.}
\label{fig:figure13}
\end{figure}

The preceding results of the parameter scan have clarified individual parameters' roles in reproducing the main features observed in the data.
Based on these analyses, we will employ the following parameters that are mostly relevant from a physical perspective:
\bqn
c_1&=&1, \nb\\
c_2&=&3, \nb\\
v_1&=&10, \nb\\
v_2&=&60, \nb\\
\sigma &=& 1,\nb\\
N_\mathrm{max}&=& 200, \nb\\
N_\mathrm{cut}&=& 150, \nb\\
L &=& 1. \label{SModelParameters}
\eqn
In the numerical calculations, we carry out 20 simulations for a given concentration $N$, each represented by a gray dot.
This gives rise to a total of 3,000 simulations.
The results are presented in Fig.~\ref{fig:figure13}.
Specifically, the results ascertain the following features of the fundamental diagram derived from the stochastic approach.
\begin{itemize}
    \item Stability and insignificant variance in the free-flow state.
    \item Scatter feature and convex shape in the congestion state.
    \item Inverse lambda and capacity drop.
\end{itemize}

The most apparent feature is the scatter exhibited by the simulated data points associated with the congested traffic state.
The almond shape in the flow variance discussed in Sec.~\ref{sec3.4} is also recognizable.
Specifically, the scatter emerges in the vicinity of $k=k_c$, and subsequently, it first increases with increasing concentration and then decreases again.
The convex shape and other features associated with theorems~\ref{flux1} and~\ref{meanandvariance} are somewhat subtle: in the left panel, the average flow remains mostly above its deterministic counterpart while constantly fluctuating.
Moreover, it is observed in the middle panel that the average flow of the stochastic model is higher than the deterministic model, forming an almond shape.
The distribution of the orange dots is rather narrow, primarily as an extension of the free-flow state.
This indicates the persistence of a small fraction of free-flow states amid congested traffic: it demonstrates the enhanced stability of the free flow in the presence of moderate stochastic noise, regarding theorem~\ref{2a3}.
As these states stretch out further into higher vehicle concentration than their deterministic counterpart, they correspond to even larger flow capacities than the peak predicted by the deterministic model.
Eventually, these states might be destabilized and merge with the congested one, indicated by gray symbols, accompanied by a lower flow capacity, triggering the capacity drop phenomenon.

\section{Concluding remarks}\label{sec5}

Capacity drop is one of the most intriguing phenomena in traffic flow analysis that has attracted continuous attention from the community.
Many studies address the problem by explicitly modeling bottlenecks in a highway section from either a microscopic or macroscopic perspective, often incorporating some degree of discontinuity into the model as a prior input.
On the other hand, the present study attributes the discontinuity and the increase in uncertainty in vehicle flow to a type of transition triggered by the underlying interplay between two factors that govern the system's stability.
From a mesoscopic perspective, the model parameters remain continuous and correspond to the transition coefficients in the Boltzmann equation.
While the concept of stability analysis is not new in traffic flow theories, the novelty of this study lies in the distinct role played by stochastic noise. 
This is particularly interesting because the deterministic counterpart of the model neither exhibits any discontinuity in the associated phase transition nor is capable of quantitatively capturing the variance of relevant physical quantities.

In this study, the stochastic noise, introduced as a specific perturbation, does not necessarily destabilize the system. 
This counterintuitive feature is embodied in the proposed stochastic traffic model, which subsequently provides a natural interpretation of the capacity drop due to the enhanced stability of the free-flow state compared to its deterministic counterpart.
Discussions about the physical implications of the results are elaborated based on several mathematical theorems, demonstrating that the present approach offers an alternative yet meaningful interpretation of the observed capacity drop and the accompanying scatter observed in the traffic data.
Specifically, the prolonged free-flow state is derived from the modified stability in the stochastic model, and an almond shape is estimated for the flow variance as a function of the concentration.
Moreover, the congested traffic phase in this model does not converge to any stationary state, as it features persistent stochastic oscillations that repeatedly cross a specific vehicle flow threshold an infinite number of times.
It is worth noting that the stochastic noise was introduced in a minimized fashion by Eq.~\eqref{EqStoAsum}. 
One might wonder about the possible outcomes if stochastic noise were introduced to other transition terms. 
Unfortunately, such alternative assumptions often lead to more complex mathematical problems where similar analytic analysis becomes infeasible.
Nonetheless, existing numerical simulations~\cite{traffic-flow-btz-lob-04} indicate the validity of the main features elaborated in the present paper in a broader context. 
The findings of this study and their possible extension to more general contexts call for further investigation.

\appendix

\section{The derivations of the theorems}\label{app2}

In this appendix, we provide sketch proofs of the theorems presented in this study.
The strategy of the derivations is partly based on Refs.~\cite{khasminskii, Gray}.
For further details, we delegate interested readers to the above literature.

\subsection{Sketch proof of Lemma \ref{invariance}}

The Lyapunov operator associated to equation (\ref{SDE}) is given by
\begin{equation}
L=\frac{\partial}{\partial t}+\left[ -c_1x+\alpha c_2(N-x)x  \right]\frac{\partial}{\partial x}+\frac{1}{2}\sigma^2\alpha^2(N-x)^2x^2\frac{\partial^2}{\partial x^2} .
\end{equation}
By taking $V(x)=\frac{1}{x}+\frac{1}{N-x}$ for $x\in(0,N)$, it is not difficult to show that $LV(x)\leq CV(x)$ for $x\in (0,N)$, where $C=c_1 \vee (c_2\alpha N)+\sigma^2\alpha^2N^2$. 
Then, by the Theorem $3.5$ of~\cite{khasminskii}, we have
\begin{equation}
\mathbb{P}(n_1(t)\in(0,N),\forall t\geq 0)=1 ,
\end{equation}
provided that $n_1(0)\in(0,N)$.

\subsection{Sketch proof of theorem \ref{2a3}}

We first define $f:\mathbb{R}\rightarrow\mathbb{R}$ by
\begin{equation}\label{4.4}
    f(x)=c_2\alpha N-c_1-c_2\alpha x-\frac{1}{2}\sigma^2(N-x)^2.
\end{equation}

By It\^o formula, we have
\begin{equation}\label{4.3}
\log(n_1(t))=\log(n_1(0))+\int_0^tf(n_1(s))ds+\int_0^t\sigma\alpha(N-n_1(s))ds .
\end{equation}
Due to $R_0^s<1$ and $\sigma^2\leq\frac{c_2}{\alpha N}$, we find
\begin{equation}
f(n_1(s))\leq c_2\alpha N -c_1-\frac{1}{2}\sigma^2\alpha^2N^2, \hspace{.2cm}\text{for $n_1(s)\in(0,N)$}.
\end{equation}
It immediately follows from Eq.~\eqref{4.3} that
\begin{equation}\label{4.5}
    \log(n_1(t))\leq \log(n_1(0))+(c_2\alpha N -c_1-\frac{1}{2}\sigma^2\alpha^2N^2)t+\int_0^t\sigma\alpha (N-n_1(s))dB(s).
\end{equation}
Dividing by $t$ on both sides of the above inequality and taking the limit $t\to \infty$, we obtain
\begin{equation}
\limsup_{t\to\infty}\frac{1}{t}\log(n_1(t))\leq c_2\alpha N -c_1-\frac{1}{2}\sigma^2\alpha^2N^2+\limsup_{t\to\infty}\frac{1}{t}\int_0^t \sigma\alpha (N-n_1(s))dB(s), \hspace{.3cm}\text{a.s.}
\end{equation}
Now, by the large number theorem for martingales, one finds
\begin{equation}
\limsup_{t\to\infty}\frac{1}{t}\int_0^t \sigma\alpha (N-n_1(s))dB(s)=0, \text{\,\, a.s.}
\end{equation}

\subsection{Sketch proof of theorem \ref{2a4}}

The function $f(x)$ defined in Eq.~\eqref{4.4} takes its maximum value $f(\hat{x})$ at $\hat{x}=\frac{\sigma^2\alpha-c_2}{\sigma^2\alpha}$. 
As $\sigma^2>\frac{ c_2}{\alpha N}\vee \frac{c_2^2}{2c_1}$, we have $\hat{x}\in(0,N)$ and
\begin{equation}
f(\hat{x})=-c_1+\frac{c_2^2}{2\sigma^2},
\end{equation}
It follows from Eq.~\eqref{4.3} that
\begin{equation}
\log(n_1(t))\leq \log(n_1(0))+f(\hat{x})t+\int_0^t\sigma\alpha(N-n_1(s))dB(s).
\end{equation}

Using the same argument as in the sketch proof of theorem {\ref{2a3}}, we have
\begin{equation}
\limsup_{t\to\infty}\frac{1}{t}\log(n_1(t))\leq f(\hat{x}), \hspace{.3cm}\text{a.s.}
\end{equation}
Finally, as $\sigma^2>\frac{c_2^2}{2c_1}$, we have $f(\hat{x})<0$.

\subsection{Sketch proof of theorem \ref{flux1}}

Let $n_1(0)\in (1, N)$. 
The equation $f(x)=0$ has a positive root given by $\xi$. 
Note that $f(0)=c_2\alpha N-c_1-\frac{1}{2}\sigma^2\alpha^2N^2>0$ due to $R_0^s>1$, and $f(N)=-c_1<0$, then $\xi\in(0,N)$.
Besides, we have
\begin{equation}\label{5.6}
    f(x)>0 \hspace{.3cm} \text{is strictly increasing on $x\in(0,0\vee \hat{x})$},
\end{equation}
\begin{equation}\label{5.7}
    f(x)>0 \hspace{.3cm} \text{is strictly decreasing on $x\in(0\vee \hat{x},\xi)$},
\end{equation}
and
\begin{equation}\label{5.8}
    f(x)<0 \hspace{.3cm} \text{is strictly decreasing on $x\in(\xi,N)$} .
\end{equation}

Suppose the statement Eq.~\eqref{first} is not valid. 
It takes $\epsilon\in(0,1)$ small enough such that $\mathbb{P}(\tilde{\Omega})>\epsilon $ with $\tilde{\Omega}=\{ \limsup_{t\to\infty} n_1(t)\leq \xi-2\epsilon \}$ and $f(0)>f(\xi-\epsilon)$. 
Therefore, there exists a random time $T>0$ on $\tilde{\Omega}$ such that
\begin{equation}\label{5.10}
    n_1(t)\leq \xi-\epsilon, \hspace{.3cm}\text{provided $t\geq T$ on $\tilde{\Omega}$ }.
\end{equation}
From Eqs.~\eqref{5.6},~\eqref{5.7}, and~\eqref{5.10}, it follows 
\begin{equation}\label{5.11}
    f(n_1(t))\geq f(\xi-\epsilon),\hspace{.3cm}\text{provided $t\geq T$ on $\tilde{\Omega}$}
\end{equation}

From Eqs.~\eqref{4.3} and~\eqref{5.11}, on $\tilde{\Omega}$, it follows that, for $t\geq T$
\begin{eqnarray}\label{5.13}
    \log(n_1(t))&=& \log(n_1(0))+\int_0^Tf(n_1(s))ds+\int_T^tf(n_1(s))ds+\int_0^t\sigma\alpha(N-n_1(s))dB(s)\\
    & \geq & \log(n_1(0))+\int_0^Tf(n_1(s))ds+f(\xi-\epsilon)(t-T)+\int_0^t\sigma\alpha(N-n_1(s))dB(s).
\end{eqnarray}
Dividing the inequality by $t$ and taking the limit $t\to\infty$, we have
\begin{equation}
\liminf_{t\to\infty}\frac{1}{t}\log(n_1(t))=\infty \hspace{.3cm}\text{on $\tilde{\Omega}$.}
\end{equation}

Therefore, $\lim_{t\to\infty}n_1(t)=\infty$ on $\tilde{\Omega}$ with $\mathbb{P}(\tilde{\Omega})>0$. 
But, this contradicts that $\mathbb{P}(n_1(t)\in(0,N),\forall t\geq 0)=1$. 
Therefore, the statement (\ref{first}) must hold.

Suppose that the statement Eq.~\eqref{second} is not true, then there is $\epsilon\in(0,1)$ small enough such that $\mathbb{P}(\hat{\Omega})>\epsilon$ with $\hat{\Omega}=\{ \liminf_{t\to\infty} n_1(t)\geq \xi+2\epsilon \}$. 
As a result, there exists a random time $T>0$ such that 
\begin{equation}\label{5.15}
    n_1(t)\geq \xi+\epsilon, \hspace{.3cm}\text{provided $t\geq T$ on $\hat{\Omega}$}.
\end{equation}

From Eqs.~\eqref{4.3} and~\eqref{5.8}, on $\hat{\Omega}$, we have
\begin{equation}\label{5.16}
    \log(n_1(t))\leq \log(n_1(0))+\int_0^Tf(n_1(s))ds+f(\xi+\epsilon)(t-T)+\int_0^T\sigma\alpha (N-n_1(s))dB(s).
\end{equation}
Subsequently, on $\hat{\Omega}$
\begin{equation}
\limsup_{t\to\infty}\frac{1}{t}\log(n_1(t))\leq f(\xi+\epsilon)<0,
\end{equation}
implies that $\lim_{t\to\infty}n_1(t)=0$ on $\hat{\Omega}$, but this contradicts (\ref{5.15}). 
Therefore, the statement (\ref{second}) holds.

\subsection{Sketch proof of theorem \ref{stationary}}

On page $107$ of Ref~\cite{khasminskii}, sufficient conditions for a unique stationary distribution are presented.
When borrowed to the context of the present study, such conditions are:

{\bf Assumption (B)} There is an open interval $(a,b)\subset (0,N)$ with the following properties:
\begin{itemize}
    \item[(B.1)] In $(a,b)$ (the diffusion) $\sigma\alpha(N-x)x$ is bounded below for a positive constant.
    \item[(B.2)] If $n_1(0)\in(0,a]\cup [b,N)$ and $\tau=\inf\{t\geq 0: n_1(t)\in (a,b)\}$, then $\tau$ is a finite stopping time almost sure, and $\sup_{n_1(0)\in K}\mathbb{E}_{n_1(0)}\tau<\infty$ for every compact subset $K\subset (0,N)$.
\end{itemize}

Obviously, for any $(a,b)\subset (0,N)$, the diffusion $\sigma\alpha(N-x)x$ satisfies (B.1). 
One can further show that the model satisfies the condition (B.2) as follows. 
For $0<a<\xi<b<N$, from Eqs.~\eqref{5.6}-\eqref{5.8}, we have
\begin{equation}\label{6.1}
    f(x)>f(0)\wedge f(a)>0 \hspace{.3cm}\text{if $0<x<a$}\hspace{.3cm}\text{and}\hspace{.3cm}f(x)\leq f(b)<0\hspace{.3cm}\text{if $b\leq x<N$}.
\end{equation}

For any $n_1(0)\in (0,a)$, due to Eqs.~\eqref{4.3} and~\eqref{6.1}, we have
\begin{equation*}
    \log(a)\geq \mathbb{E}_{n_1(0)}(\log(n_1(\tau\wedge t)))\geq  \log(n_1(0))+(f(0)\wedge f(a))\mathbb{E}_{n_1(0)}(\tau \wedge  t),\hspace{.3cm}\forall t\geq 0.
\end{equation*}
Taking $t\to\infty$, we obtain
\begin{equation}\label{6.2}
    \mathbb{E}_{n_1(0)}(\tau)\leq \frac{\log(a/n_1(0))}{f(0)\wedge f(a)}, \forall n_1(0)\in(0,a).
\end{equation}
For any $n_1(0)\in (b,N)$, 
\begin{equation}
\log(b)\leq \mathbb{E}_{n_1(0)}(\log(n_1(\tau \wedge t)))\leq \log(n_1(0))+(f(0)\wedge f(a))\mathbb{E}_{n_1(0)}(\tau \wedge t),\hspace{.3cm}\forall t\geq 0.
\end{equation}
Taking $t\to\infty$, we have
\begin{equation}\label{6.3}
\mathbb{E}_{n_1(0)}(\tau)\leq \frac{\log(N/b)}{|f(b)|}, \forall n_1(0)\in (b,N).
\end{equation}
Subsequently, $n_1(0)\in (a,b)$ implies that $\tau\equiv 0$, by Eqs.~\eqref{6.2} and~\eqref{6.3}, the condition (B.2) holds. 

\subsection{Sketch proof of theorem \ref{meanandvariance}}

To proceed, the following theorem is crucial:
\begin{Thm}\label{ergodic}
Suppose that $\pi$ is the stationary distribution of the process $n_1(t)$. 
Let $g(x)$ be a function integrable concerning the measure $\pi$. 
Then, for given $n_1(0)\in(0,N)$, we have
\begin{equation}\label{ergodic1}
    \mathbb{P}\left[ \lim_{t\to\infty}\int_0^t g(n_1(s))ds=\int_0^Ng(y)\pi(dy)     \right]=1
\end{equation}
\end{Thm}

\begin{proof}
See page $110$ of Ref.~\cite{khasminskii}.
\end{proof}

Note that $\pi$ is a measure on $\mathcal{B}(0, N)$, and subsequently, $\pi$ has nonvanishing moments of all orders. 
In particular, $\mu=\int_0^Ny\pi(dy)$ and $\gamma=\int_0^Ny^2\pi(dy)-\mu^2$ are real numbers. 
It follows from Eq.~\eqref{SDE} that
\begin{equation}
n_1(t)=n_1(0)+\int_0^t n_1(s)\left[ c_2\alpha N-c_1 -c_2\alpha n_1(s) \right]ds+ \int_0^t \sigma\alpha (N-n_1(s))dB(s) .
\end{equation}
When divided $t$ while taking the limit $t\to\infty$, theorem~\ref{ergodic} implies
\begin{equation}\label{6.6}
0=(c_2\alpha N-c_1)\mu-c_2\alpha (\gamma+\mu^2) .
\end{equation}

From Eq.~\eqref{4.3}, we have
\begin{equation}\label{eq}
    \lim_{t\to\infty}\frac{\log(n_1(t))}{t}=c_2\alpha N-c_1-\frac{1}{2}\sigma^2\alpha^2N^2-(c_2\alpha-\sigma^2N)\mu-\frac{1}{2}\sigma^2\alpha^2(\gamma+\mu^2), \hspace{.2cm}\text{a.s}.
\end{equation}
Subsequently, by theorems~\ref{flux1} and Eq.~\eqref{eq}, we have
\begin{equation}
\lim_{t\to\infty}\frac{\log(n_1(t))}{t}=0, \hspace{.3cm}\text{a.s}.
\end{equation}

By noting that $c_2\alpha N-c_1-\frac{1}{2}\sigma^2\alpha^2N^2=(R_0^s-1)c_1$, one finds
\begin{equation}\label{6.8}
    0=(R_0^s-1)c_1-(c_2\alpha -\sigma^2\alpha^2N)\mu-\frac{1}{2}\sigma^2\alpha^2(\gamma+\mu^2) .
\end{equation}
Substituting Eq.~\eqref{6.6} into Eq.~\eqref{6.8}, we obtain
\begin{equation}
0=(R_0^s-1)c_1-(c_2\alpha-\sigma^2N)\mu-\frac{\sigma^2\alpha^2(c_2\alpha N-c_1)\mu}{2c_2\alpha} ,
\end{equation}
and as a result, Eq~\eqref{mean} holds. 

From Eq.~\eqref{6.6}, we have
\begin{equation}
(\gamma+\mu^2)=\frac{\mu(c_2\alpha N-c_1)}{c_2\alpha} ,
\end{equation}
which gives
\begin{equation}
\gamma=\frac{\mu(c_2\alpha N-c_1)}{c_2\alpha}-\mu^2.
\end{equation}

\section{Numerical validation for some of the theorems}\label{app3}

\subsection{Numerical validation for theorem~\ref{2a3}}\label{app3a1}

In order to corroborate the results of theorem~\ref{2a3}, simulations are performed for Eq.~\eqref{SDE}.
In the calculations, we consider 1,000 different parameter sets consisting of $N$, $\sigma$, $c_1$, and $c_2$, which are selected as follows. 
An integer $N$ is drawn from the interval $[50, 150]$ for the total vehicle number. 
Two real values, $c_1$ and $c_2$, are drawn from the uniform distribution $\mathrm{U}(1,6)$ for the transition coefficients. 
Similarly, $\sigma$ is generated from the uniform distribution $\mathrm{U}(0.2, 1.2)$. 
For each of these combinations, 100 simulations are performed.
In the simulations, we deliberately suppressed the results near the points $R_0^s = 1$ and $\sigma^2 = \frac{c_2}{\alpha N}$ due to convergence issues that significantly raised the computational time.
For the remaining parameters, we adopt $N_\mathrm{max} = 200$, $L = 1$, $v_1 = 10$, and $v_2 = 60$. The variable $n_1$ evolves in time according to the stochastic differential equation Eq.~\eqref{SDE}, considering $n_1(0)$ taken from the uniform distribution $\mathrm{U}(1, N)$.

For each of the 100,000 simulations, 30,000 time steps are considered, ranging from 0 to 30. 
As time $t$ evolves, one constantly verifies the following inequality
\begin{equation}
\label{cond}
n_1(t) < \exp([-\epsilon \,\, t]) .
\end{equation}
Once it is satisfied for $t=t_s$ and remains valid until $t_f \equiv 3t_s$, it is understood that the above inequality is satisfied for any $t> t_s$ from a numerical perspective.
The time instant $t_s$ will be thus denoted as ``the time of convergence,'' indicating how fast a stable free-flow state is achieved.
We adopted the value $\epsilon=0.1$ for these simulations. 

The distribution of the time of convergence $t_s$ in the parameter space is shown as a color map in the left panel of Fig.~\ref{fig:figure4}, where the value $t_s$ is the average among 100 simulations for a given set of parameters $N$, $\sigma$, $c_1$, and $c_2$.
The color map indicates that the system takes longer to converge to the free flow state as the system starts to saturate the stability bound dictated by theorem~\ref{2a3}.

The right panel of Fig.~\ref{fig:figure4} shows the cumulative distribution of $t_s$.
We note that the resulting distribution is not normal.
Instead, it increases monotonically with $t_s$ and seems to converge at a significant convergence time.
This can be understood as not only the stability of a free-flow state becomes increasingly more difficult as the system approaches the bound, but these states also take up increasingly more weight in the parameter space.
As they are located near the congestion phase while featuring more flow fluctuations, these states might be challenging to distinguish from those related to traffic congestion.

\begin{figure}[!htb]
\begin{tabular}{cc}
\begin{minipage}{250pt}
\centerline{\includegraphics[width=1.0\textwidth]{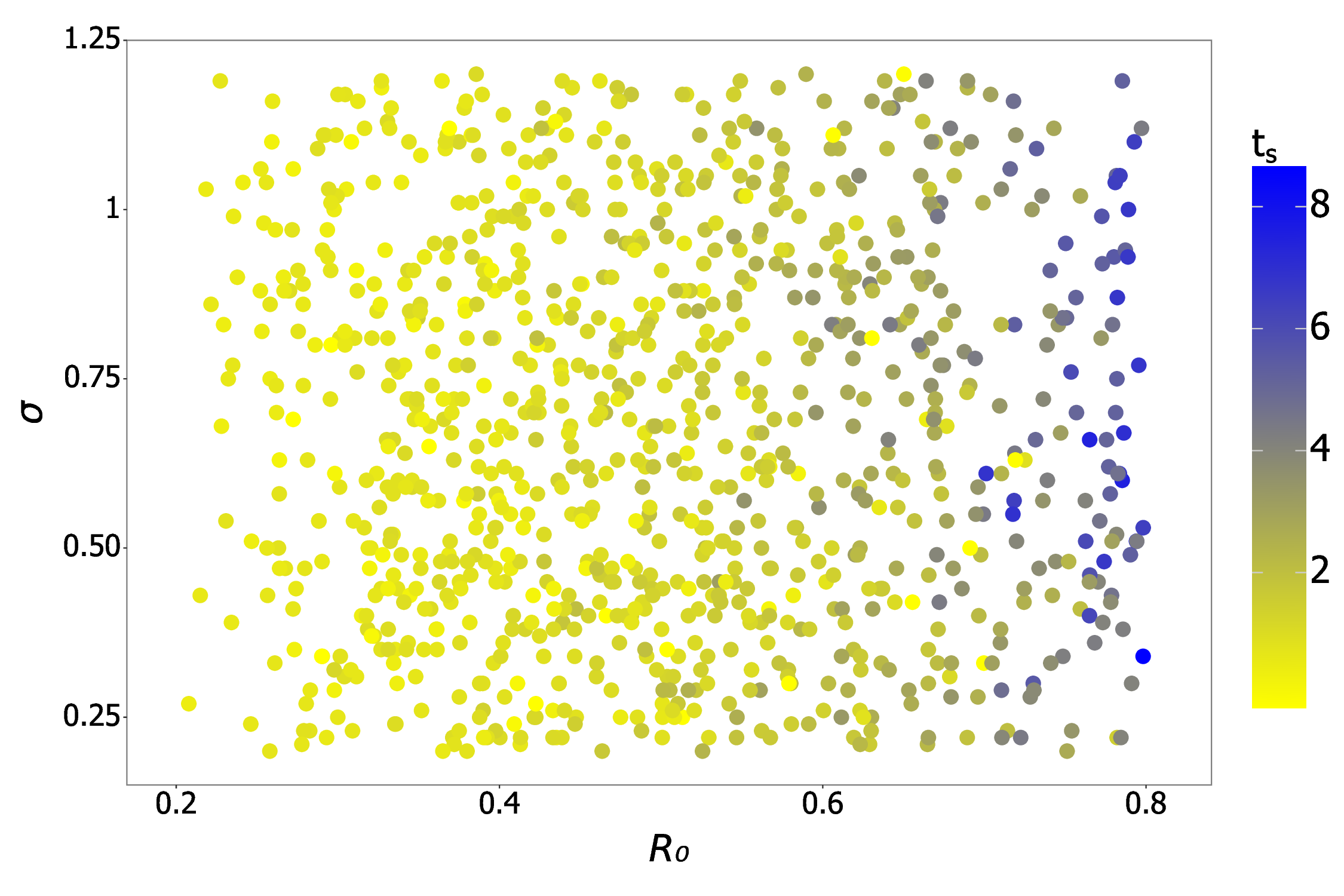}} \scriptsize{(a)}
\end{minipage}
&
\begin{minipage}{250pt}
\centerline{\includegraphics[width=1.0\textwidth]{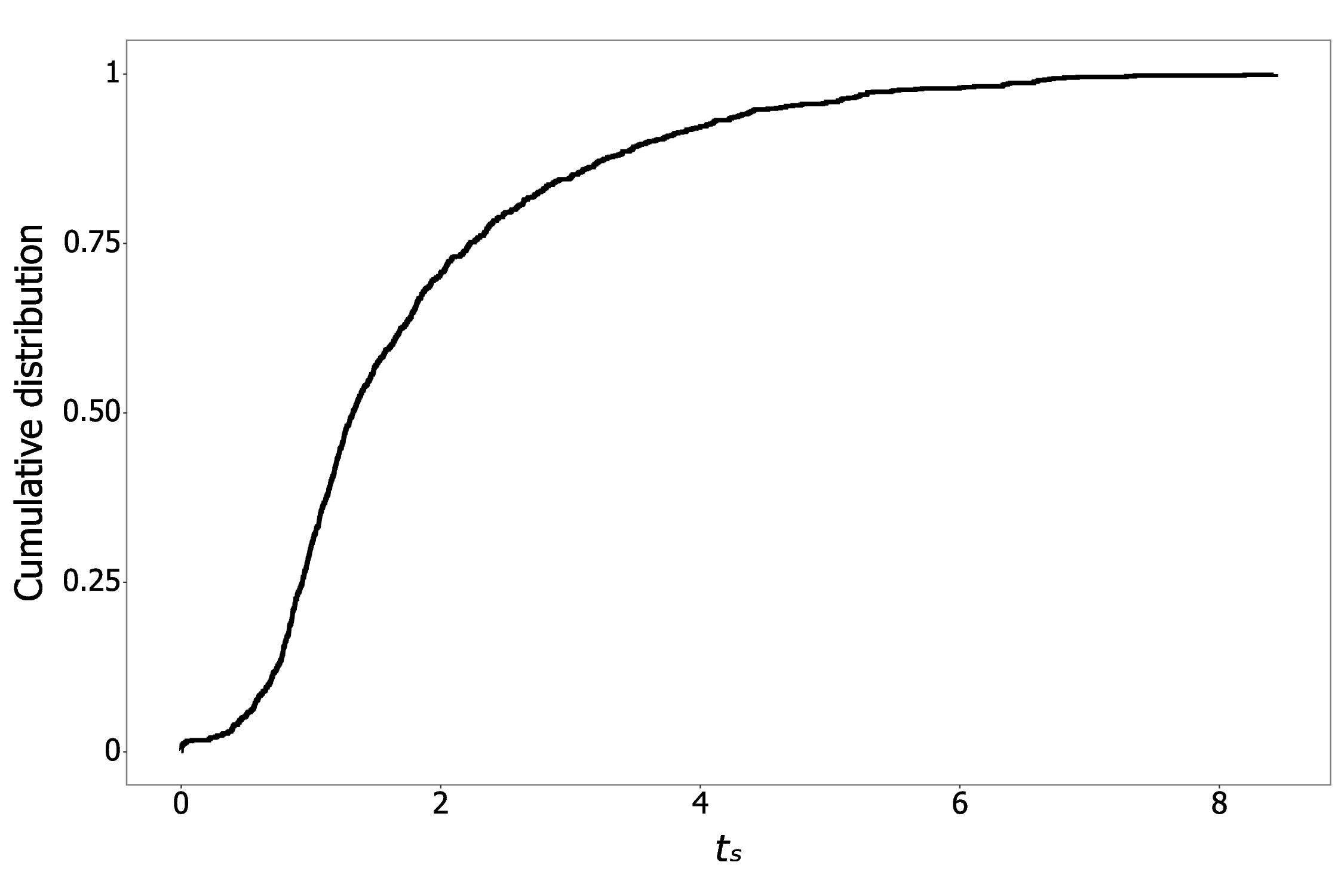}} \scriptsize{(b)}
\end{minipage} 
\end{tabular}
\renewcommand{\figurename}{Fig.}
\caption{(a) A color map of the time of convergence $t_s$ as a function of $R_0$ and $\sigma$.
Lighter color (more yellowish) dots represent the states with a smaller $t_s$, while the darker (blueish color) ones correspond to those with more significant $t_s$. 
(b) The corresponding cumulative distribution of $t_s$, obtained using 1,000 parameter combinations.}
\label{fig:figure4}
\end{figure}

Furthermore, simulations are performed to compute the confidence intervals of the concentration $n_1$ evaluated at two different time instants, $t_s= 13.5$ and $98.5$.
Specifically, 2,000 simulations are carried out for two small intervals $t \in [12.5, 14.5]$ and $[97.5, 99.5]$.
In the calculations, the model parameters $N$, $c_1$, $c_2$, $\sigma$, and $N_\mathrm{max}$ are chosen following the same scheme described above.
For each time interval, we consider a sample size of 100 and analyze the following percentiles: $p(0.50)$, $p(0.75)$, $p(0.85)$, $p(0.90)$, $p(0.95)$, $p(0.99)$, $p(0.995)$, and $p(0.999)$.
Theoretically, if many samples were taken and intervals calculated, the percentile $p(k)$ indicates that $k$\% of such intervals would contain the true value.
The interval $p(k)$ is obtained by evaluating the sample mean, standard deviation, and the corresponding t-value.
The results are presented in Tab.~\ref{Table1}. 
Since the value zero is included in all the confidence intervals, the results agree with theorem~\ref{2a3}. 
To be more precise, the obtained confidence intervals do not provide any evidence to reject the null hypothesis of the theorem's validity.
Also, as expected, the amplitude of the interval decreases with increasing time.

\begin{table}[!h]
\centering
\caption{Confidence intervals of $n_1$ for different percentiles evaluated at two different time instants.
The simulations are carried out using the model parameters when both conditions of the theorem~\ref{2a3} are satisfied.}
\begin{tabular}{c | c c c c | c c c c}
\hline \hline
\multirow{3}{*}{percentiles} &  \multicolumn{4}{c|}{time interval $t \in [12.5, 14.5]$} &  \multicolumn{4}{c}{time interval $t \in [97.5, 99.5]$}\\
\cline{2-9}
& point &  \multicolumn{3}{c|}{confidence interval} & point &  \multicolumn{3}{c}{confidence interval} \\
 & estimation & lower limit & upper limit & amplitude & estimation & lower limit & upper limit & amplitude  \\
\hline 
p(0.50) & 0.5580 & -0.3808 & 1.4968 & 1.8777 &
0.1816 & -0.0801 & 0.4432 & 0.5233 \\
p(0.75) & 0.7148 & -0.4930 & 1.9225 & 2.4155 &
0.5837 & -0.3582 & 1.5256 & 1.8838 \\
p(0.85) & 0.7944 & -0.5504 & 2.1392 & 2.6896 &
0.7027 & -0.4438 & 1.8492 & 2.2930 \\
p(0.90) & 0.8267 & -0.5706 & 2.2241 & 2.7947 &
0.7803 & -0.4986 & 2.0593 & 2.5579 \\
p(0.95) & 0.8865 & -0.6074 & 2.3804 & 2.9878 &
0.8534 & -0.5465 & 2.2534 & 2.7998 \\
p(0.99) & 0.9614 & -0.6511 & 2.5740 & 3.2251 &
0.9772 & -0.6181 & 2.5725 & 3.1906 \\
p(0.995) & 0.9814 & -0.6643 & 2.6271 & 3.2914 &
0.9944 & -0.6212 & 2.6101 & 3.2313 \\
p(0.999) & 1.0139 & -0.6902 & 2.7180 & 3.4082 &
1.0314 & -0.6440 & 2.7068 & 3.3508 \\
\hline \hline
\end{tabular}
\label{Table1}
\end{table}

\subsection{Numerical validation for theorem~\ref{2a4}}

In order to confirm the results of theorem~\ref{2a4}, following a similar strategy, 1,000 different parameter sets for $N$, $\sigma$, $c_1$, and $c_2$ are considered, and 100 simulations are performed for each given parameter setup.
Similar to the previous case, in the simulations, we deliberately suppressed the results near the bound $\sigma^2 = \frac{c_2}{\alpha N} \vee \frac{c_2^2}{2c_1}$ due to overwhelming computational costs related to slow convergence issues.
For each of the 100,000 simulations, 30,000 time steps are considered, ranging from 0 to 30. 
As time $t$ evolves, the inequality Eq.~(\ref{cond}) is consistently satisfied. 
Again, once it is satisfied for $t=t_s$ and remains valid until $t_f \equiv 3t_s$, it is understood that this inequality is satisfied for any $t> t_s$ from a numerical perspective.
For these simulations, we also adopted the constant $\epsilon=0.1$. 

The distribution of the time of convergence $t_s$ according to $\frac{c_2}{\alpha N} \vee \frac{c_2^2}{2c_1}$ and $\sigma$ values is presented as a color map in the left panel of Fig.~\ref{fig:figure5}, where the value $t_s$ is the average among 100 simulations for a given set of parameters $N$, $\sigma$, $c_1$, and $c_2$.
The right panel of this figure shows the cumulative distribution of $t_s$.
These results are consistent with those observed in the previous section.

\begin{figure}[!htb]
\begin{tabular}{cc}
\begin{minipage}{250pt}
\centerline{\includegraphics[width=1.0\textwidth]{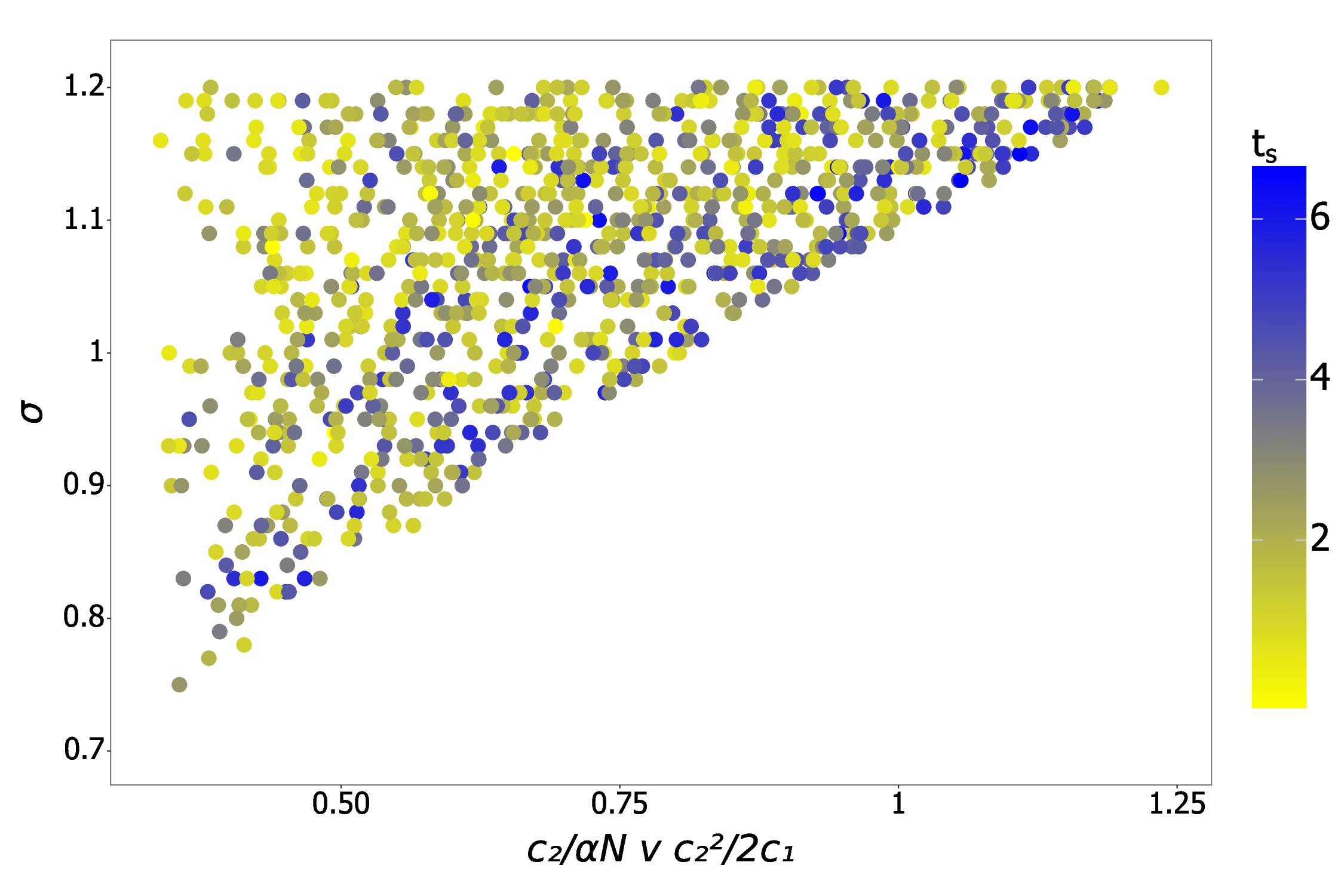}}  \scriptsize{(a)}
\end{minipage}
&
\begin{minipage}{250pt}
\centerline{\includegraphics[width=1.0\textwidth]{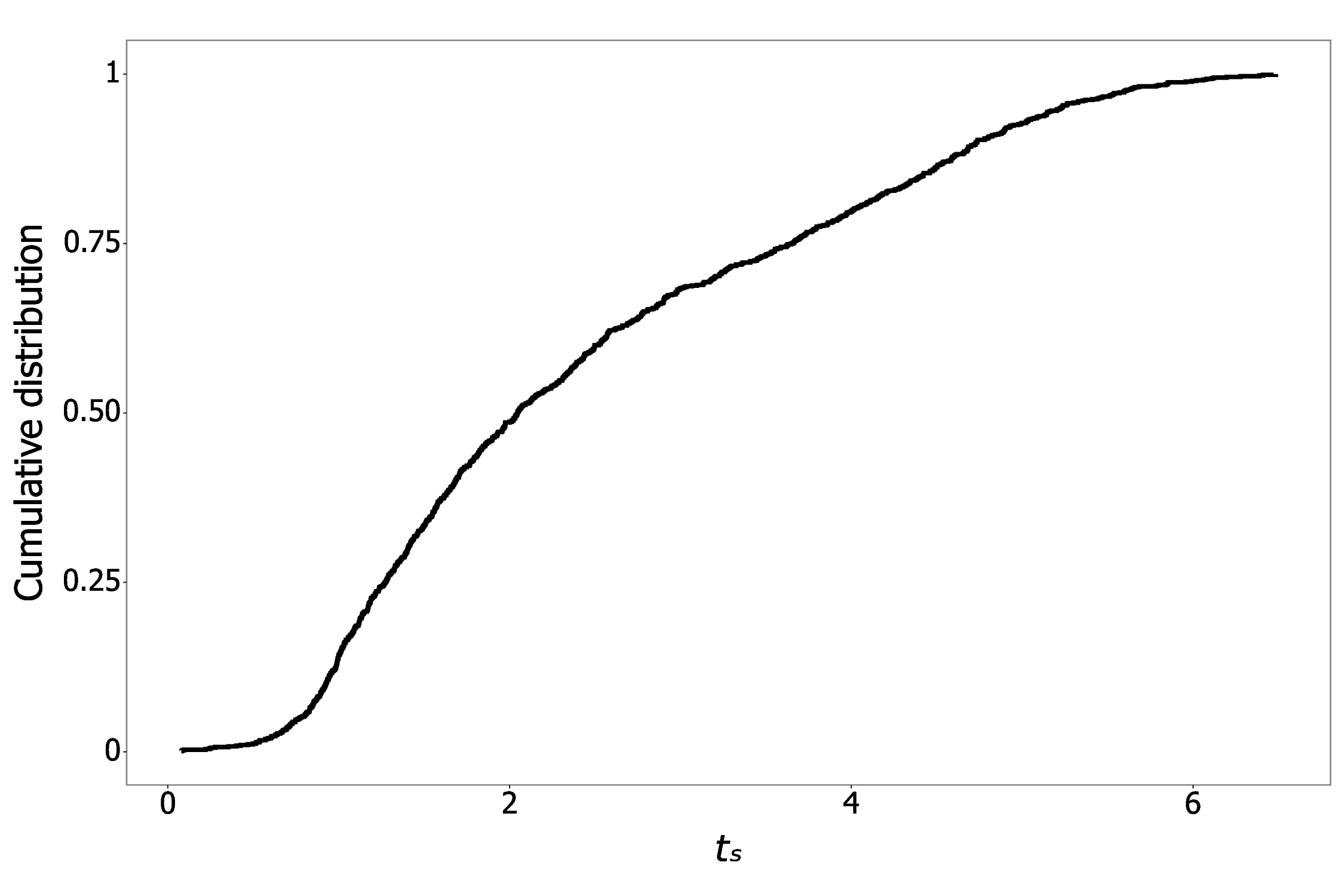}}  \scriptsize{(b)}
\end{minipage} 
\end{tabular}
\renewcommand{\figurename}{Fig.}
\caption{(a) A color map of the time of convergence $t_s$ as a function of $\frac{c_2}{\alpha N} \vee \frac{c_2^2}{2c_1}$ and $\sigma$. Lighter color (more yellowish) dots represent the states with a smaller $t_s$, while the darker (blueish color) ones correspond to those with more significant $t_s$. 
(b) The corresponding cumulative distribution of $t_s$, obtained using 1,000 parameter combinations.}
\label{fig:figure5}
\end{figure}

Also, simulations are carried out to calculate the confidence intervals for the percentiles obtained for two time instants $t_s=13.5$ and $t_s=98.5$, again, implemented in terms of small time intervals, $t \in [12.5, 14.5]$ and $t \in [97.5, 99.5]$.
For each arbitrary parameter set consisting of randomly chosen $N$, $c_1$, $c_2$, $\sigma$, and $N_\mathrm{max}$ satisfying the condition of theorem~\ref{2a4}, different percentiles are calculated.
The confidence intervals for each percentile are obtained from a sample size of 100 independent simulations, and the results are presented in Tab.~\ref{Table2}. 
The obtained results agree with theorem \ref{2a4}, as zero is contained in all the calculated confidence intervals. 
In addition, as expected, the amplitude of the interval decreases as time increases.

\begin{table}[!h]
\centering
\caption{Confidence intervals of $n_1$ for different percentiles evaluated at two different time instants.
The simulations are carried out using the model parameters when both conditions of the theorem~\ref{2a4} are satisfied.}
\begin{tabular}{c | c c c c | c c c c}
\hline \hline
\multirow{3}{*}{percentiles} &  \multicolumn{4}{c|}{time interval $t \in [12.5, 14.5]$} &  \multicolumn{4}{c}{time interval $t \in [97.5, 99.5]$}\\
\cline{2-9}
& point &  \multicolumn{3}{c|}{confidence interval} & point &  \multicolumn{3}{c}{confidence interval} \\
 & estimation & lower limit & upper limit & amplitude & estimation & lower limit & upper limit & amplitude  \\
\hline 
p(0.50) & 0.8800 & -0.5898 & 2.3497 & 2.9396 &
0.1894 & -0.1493 & 0.5282 & 0.6775 \\
p(0.75) & 1.1383 & -0.6768 & 2.9534 & 3.6303 &
0.8799 & -0.5898 & 2.3496 & 2.9394 \\
p(0.85) & 1.2841 & -0.6958 & 3.2641 & 3.9599 &
1.0865 & -0.6663 & 2.8394 & 3.5056 \\
p(0.90) & 1.3760 & -0.6943 & 3.4463 & 4.1406 &
1.2128 & -0.6891 & 3.1146 & 3.8037 \\
p(0.95) & 1.5180 & -0.6366 & 3.6725 & 4.3091 &
1.3760 & -0.6943 & 3.4463 & 4.1406 \\
p(0.99) & 1.6616 & -0.6074 & 3.9306 & 4.5380 &
1.6091 & -0.6171 & 3.8353 & 4.4525 \\
p(0.995) & 1.7051 & -0.6126 & 4.0228 & 4.6355 &
1.6616 & -0.6074 & 3.9306 & 4.5380 \\
p(0.999) & 1.7658 & -0.6208 & 4.1524 & 4.7733 &
1.7509 & -0.6223 & 4.1241 & 4.7464 \\
\hline \hline
\end{tabular}%
\label{Table2}%
\end{table}%

\subsection{Numerical validation for theorem~\ref{meanandvariance}}

To assess theorem~\ref{meanandvariance}, 300 randomly generated parameter sets for $N$, $c_1$, $c_2$, and $\sigma$ are utilized.
Due to convergence issues that significantly raise the computational costs, we deliberately avoid the choices near the boundary $R_0^s = 1$.
The remaining parameters take the following values: $N_\mathrm{max} = 200$, $L=1$, $v_1=10$, and $v_2=60$. 
In our simulations, one considers 30,000 time steps, ranging from 0 to 30.
For a given parameter set, 100 simulations are carried out, where the initial value of $n_1(0)$ is drawn from the uniform distribution $\mathrm{U}(1, N)$.
The simulation is terminated at $t_s=29.25$, implemented by a short time interval $t \in [29,29.5]$, for which it is understood that the asymptotical state described by the theorem has been reached.

Subsequently, the mean and the variance of $n_1$ are evaluated and compared with those estimated by theorem~\ref{meanandvariance}.
The comparisons are carried out by evaluating the ratio between the values obtained by the simulation and those obtained using Eqs.~(\ref{mean}) and (\ref{variance}).
These results are shown in Fig.~\ref{fig:figure6}.
Tab.~\ref{Table3} gives a more qualitative summary in terms of percentiles. 
These results provide no evidence to reject the null hypothesis of the theorem's validity.

\begin{figure}[!htb]
\begin{tabular}{cc}
\begin{minipage}{250pt}
\centerline{\includegraphics[width=1.0\textwidth]{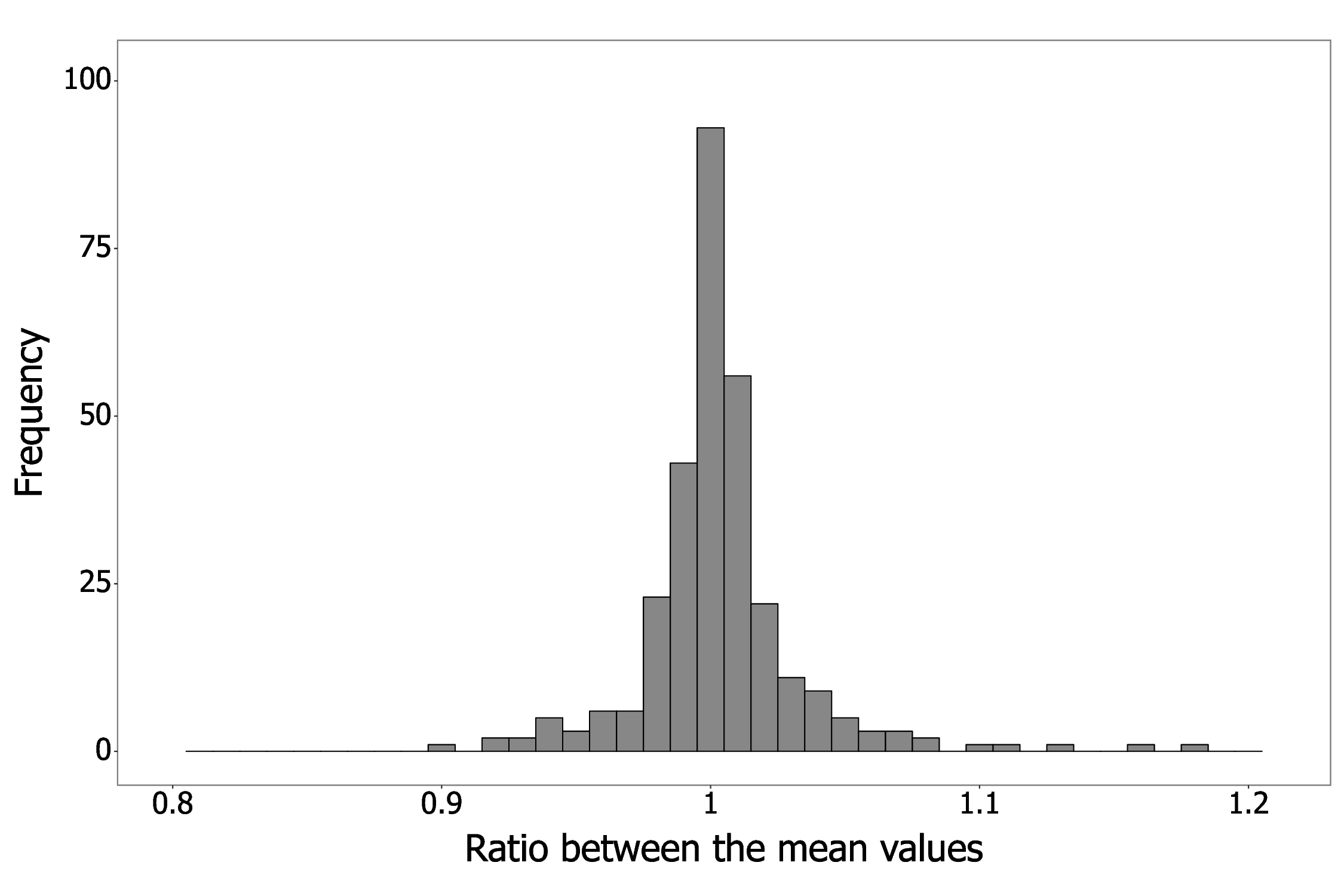}}  \scriptsize{(a)}
\end{minipage} &
\begin{minipage}{250pt}
\centerline{\includegraphics[width=1.0\textwidth]{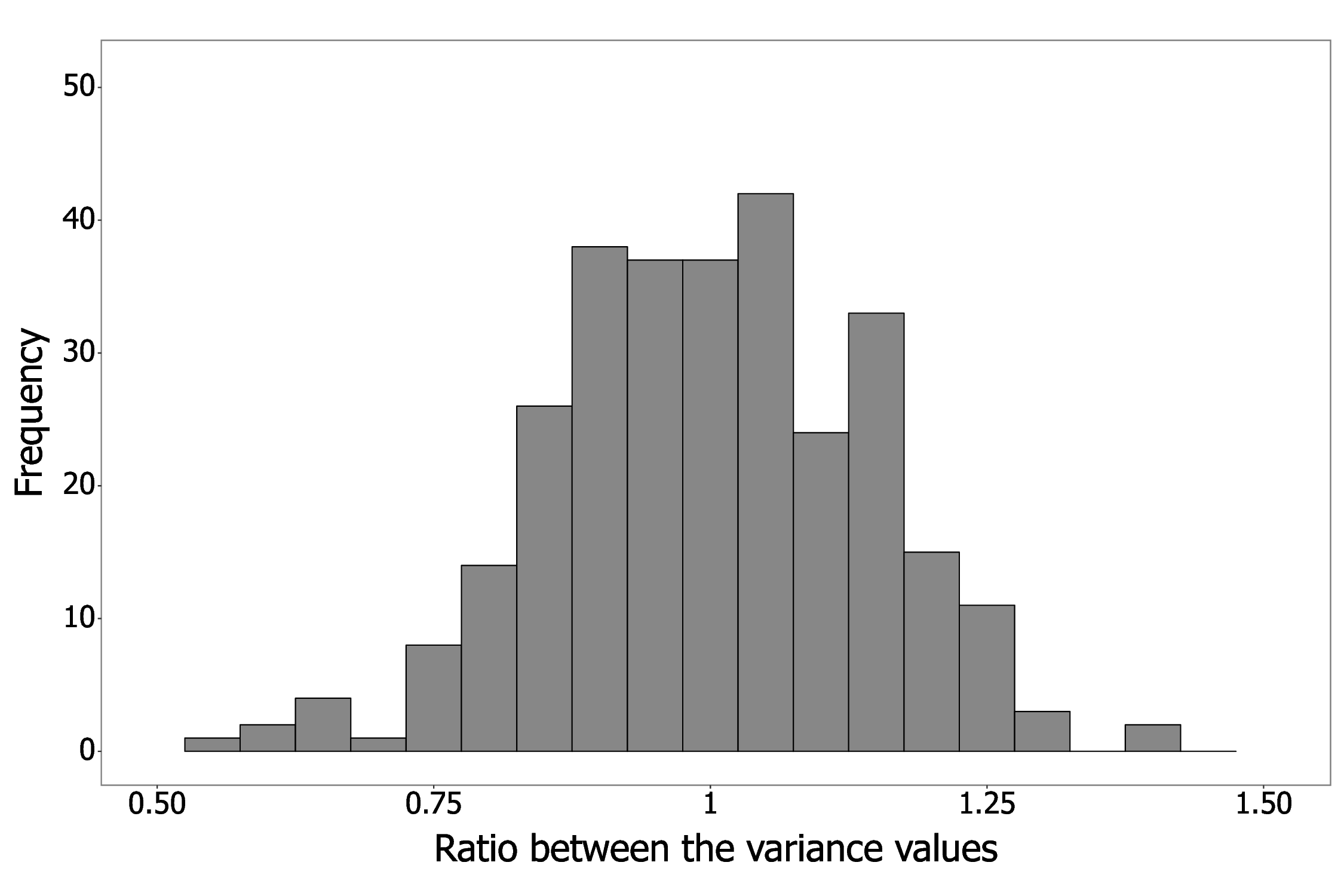}}  \scriptsize{(b)}
\end{minipage} \\
\end{tabular}
\renewcommand{\figurename}{Fig.}
\caption{The ratios between the values obtained by the simulation and those obtained using theorem~\ref{meanandvariance} for the mean (left panel) and variance (right panel).}
\label{fig:figure6}
\end{figure}

\begin{table}[!h]
\centering
\caption{A summary of the ratios between the values estimated by the simulations and those calculated according to theorem~\ref{meanandvariance}.}
\begin{tabular}{c c c}
\hline \hline
\textbf{summary measures} & \textbf{ratio of means}  & \textbf{ratio of variances} \\
\hline
mean           & 1.0033 & 1.0031 \\
std.deviation  & 0.0297	& 0.1555 \\
minimum        & 0.9007  & 0.5656 \\
p(0.25)        & 0.9920  & 0.8985 \\
p(0.50)        & 1.0017  & 1.0030 \\ 
p(0.75)        & 1.0106  & 1.1009 \\
maximum        & 1.1821  & 1.9100 \\
\hline \hline
\end{tabular}
\label{Table3}%
\end{table}%

\subsection{Numerical validation for theorem~\ref{flux1}}

To illustrate the validity of theorem~\ref{flux1}, simulations are carried out using an arbitrarily chosen model parameter set $N = 150$, $c_1 = 1$, $c_2 = 3$, $\sigma = 1$, $N_\mathrm{max} = 200$, $L=1$, $v_1=10$, and $v_2=60$.
For these parameters, one finds $R_0^s = 4.5$, which satisfies the criterion for the theorem's validity. 
The specific vehicle occupation $\xi$, associated with the crossings and given by Eq.~(\ref{xi}), is 132.2876.
The initial condition of the simulations $n_1(0)$ is taken from the uniform distribution $\mathrm{U}(1, N)$.

The results are presented in Fig.~\ref{fig:figure7}, considering 20 simulations and 30,000 time steps, ranging from 0 to 30.
In the left panel, we show the accumulated count for the number of crossings (w.r.t. $\xi$) as a function of time. 
To illustrate the results more transparently, we use different colors to represent different simulations.
It is apparent that the number of crossings increases monotonically in time with a roughly constant rate, in agreement with theorem~\ref{flux1}.
In the right panel of Fig.~\ref{fig:figure7}, we show the value of $n_1$ as a function of time.
The feature of the trajectories indicates that the temporal dependence of these states is not stationary.
The results shown in both panels are confirmed (not shown here) by carrying out a more significant number of simulations over a more extensive period.

\begin{figure}[!htb]
\begin{tabular}{cc}
\begin{minipage}{250pt}
\centerline{\includegraphics[width=1.0\textwidth]{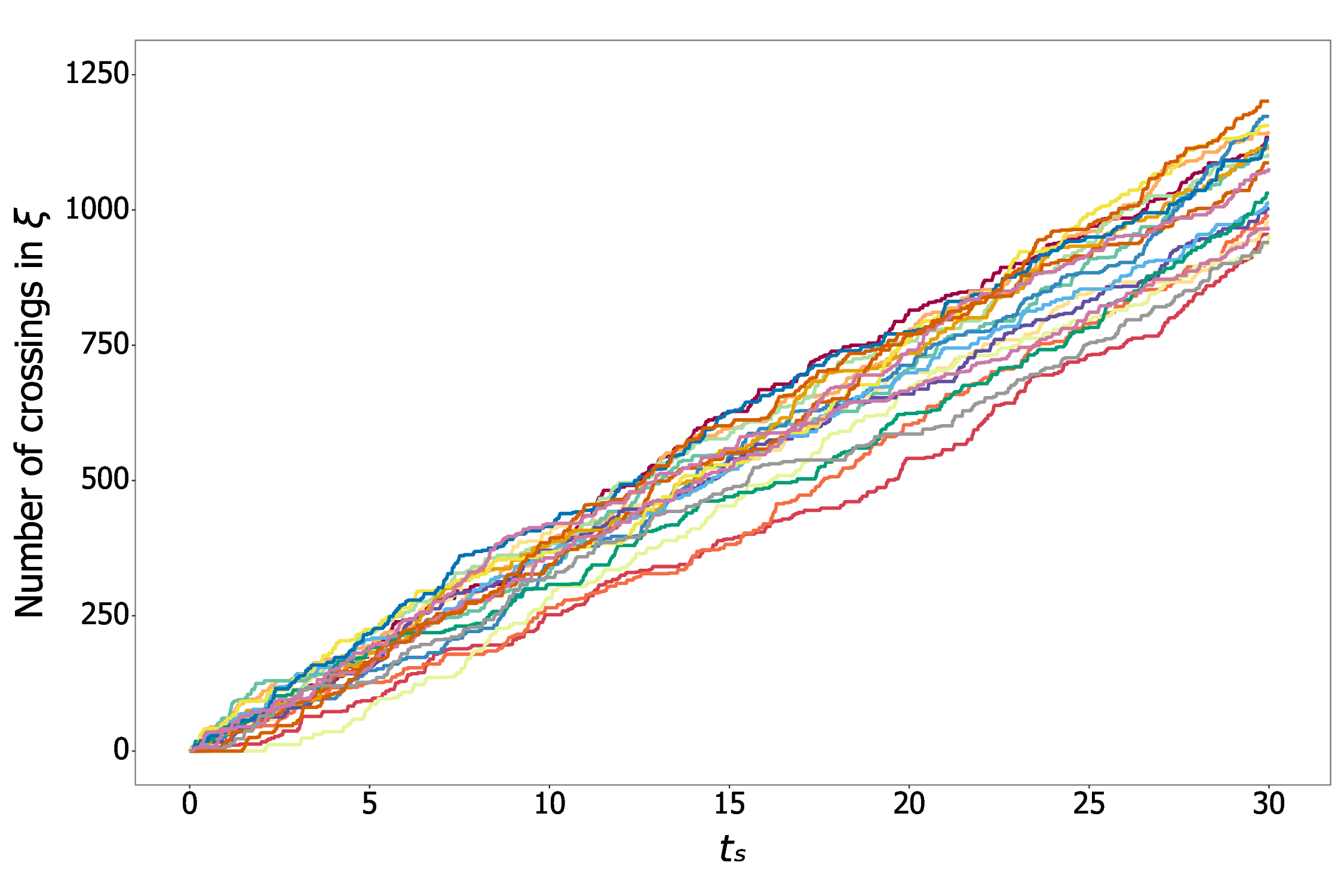}} \scriptsize{(a)}
\end{minipage} 
&
\begin{minipage}{250pt}
\centerline{\includegraphics[width=1.0\textwidth]{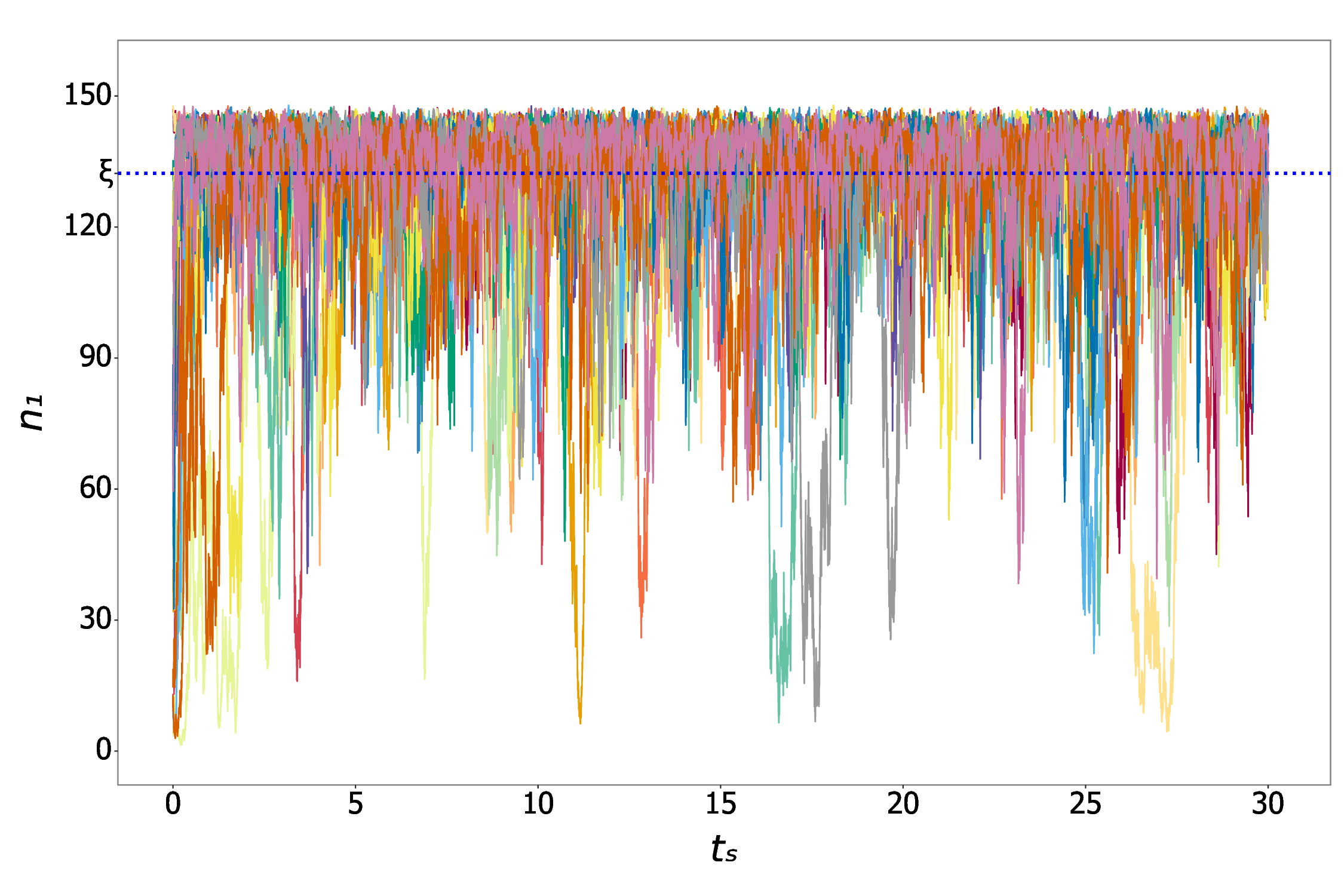}} \scriptsize{(b)}
\end{minipage} 
\end{tabular}
\renewcommand{\figurename}{Fig.}
\caption{(a) The cumulative number of crossings as a function of time. 
(b) The values of vehicle accumulation $n_1$ as a function of time.
The results are obtained by considering 20 simulations, indicated by different colors.}
\label{fig:figure7}
\end{figure}

\section{Effects of the remaining model parameters}\label{app4}

\subsection{Effect of time adopted in the calculations of the flow-concentration curves}

In calculating the fundamental diagrams shown in Figs.~\ref{fig:figure18} and~\ref{fig:figure13} in the main text, the relevant quantities are evaluated at the time instant $t_s\sim 26$.
To confirm the robustness of the obtained result, different values of $t_s$ are utilized for evaluating the flow-concentration curve, and the results are shown in Fig.~\ref{fig:figure8}.
In the figure, the gray points represent the individual estimated values, the solid red line corresponds to their mean, and the dashed blue line indicates that of the deterministic model.
In the calculations, one considers the parameters $N_\mathrm{max} = 200$, $L = 1$, $v_1 = 10$, $v_2 = 60$, $c_1 = 1$, $c_2 = 3$, $\sigma = 1$, and for different $t \in \{5, 10, 15, 20, 25, 30\}$.
The initial condition $n_1(0)$ is drawn from the uniform distribution $\mathrm{U}(1, N)$.
One carries out 20 simulations for each given concentration $k$, ranging from 1 to 150 and incremented by 1. 
The deterministic counterpart is obtained by repeating one simulation assuming $\sigma = 0$.

As $t$ increases, it is observed that the plots converge in terms of the distribution of the gray symbols and the variance of the red curves.
The average flow in most of these plots is observed to be close to or above the value governed by the deterministic model. 
The difference observed between the plots with smaller values of $t$ could be primarily attributed to the randomness of the initial condition $n_1(0)$, as the latter influences the estimation more significantly during the earlier stage.
By Fig.~\ref{fig:figure8} and considering the estimated time of convergence $t_s$ in Appx.~\ref{app3a1}, the value $t_s\sim 26$ is taken for our calculations.

\begin{figure}[!htb]
\begin{minipage}{250pt}
\centerline{\includegraphics[width=1.8\textwidth]{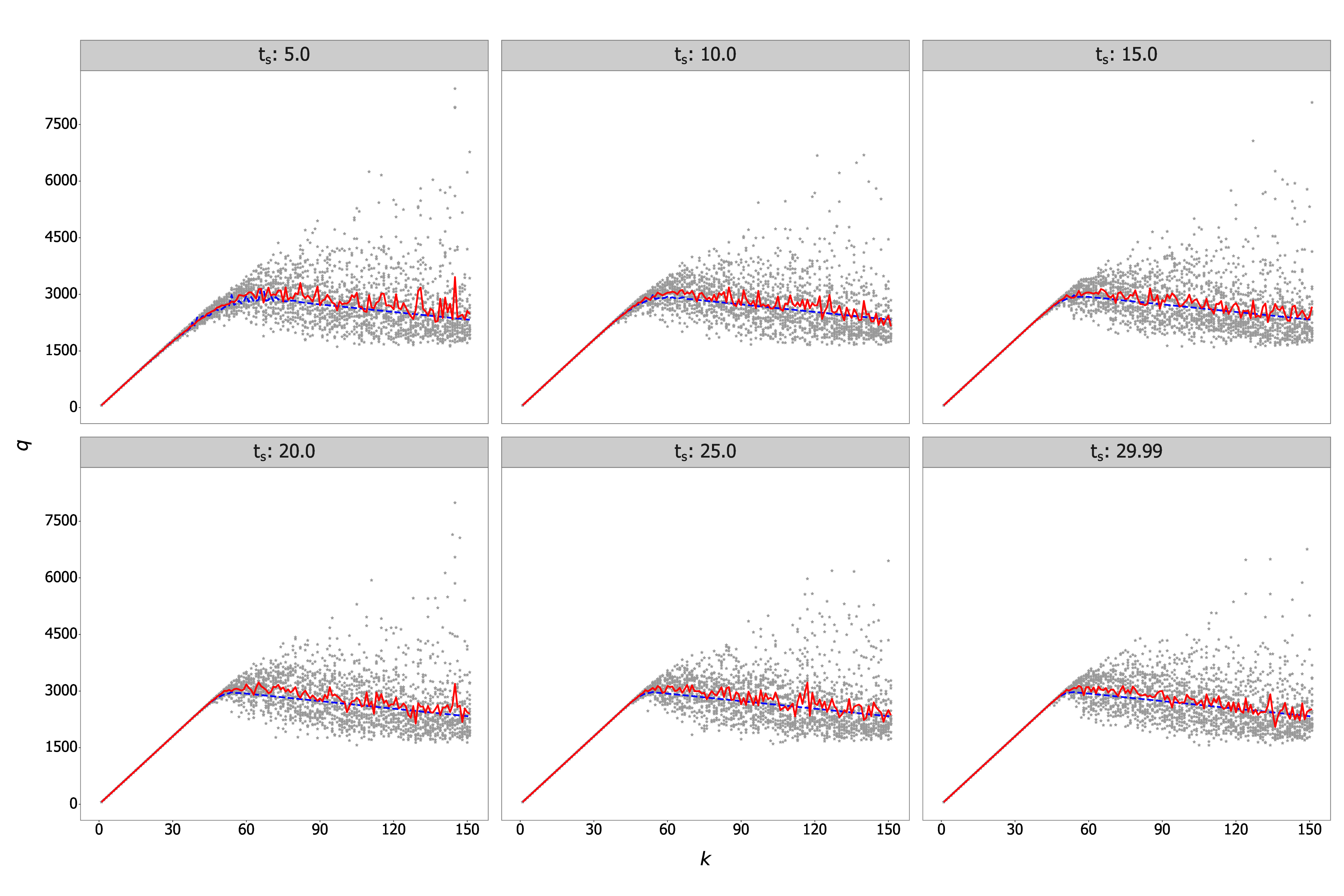}}
\end{minipage}
\renewcommand{\figurename}{Fig.}
\caption{The flow-concentration curves evaluated at different time instants. 
The convention used in the plots is the same as in Fig.~\ref{fig:figure18}.
In the calculations, we vary the value $t$ and adopt the following model parameters $v_1 = 10$, $v_2 = 60$, $L = 1$, $N_\mathrm{max} = 200$, $c_1 = 1$, $c_2 = 3$, $\sigma = 1$, and $N_\mathrm{cut} = 150$.}
\label{fig:figure8}
\end{figure}

\subsection{Effect of the strength of the stochastic noise in the flow-concentration curve}

Here, we turn to assess the impact of $\sigma$.
In Fig.~\ref{fig:figure9}, we show the resulting fundamental diagrams using different values of stochastic noise.
One considers 20 simulations for a given concentration value, which takes 150 discrete values on the axis.
The initial condition $n_1(0)$ is drawn from the uniform distribution $\mathrm{U}(1, N)$ and $t_s=30$. 
The calculations are carried out using the parameters $v_1 = 10$, $v_2 = 60$, $L = 1$, $N_\mathrm{max} = 200$, $c_1 = 1$, $c_2 = 3$, $t_s= 26$, and $N_\mathrm{cut} = 150$.
In the upper row of Fig.~\ref{fig:figure9}, one considers a relatively minor stochastic noise with $\sigma=0.5$, while in the bottom row, one adopts a more significant one with $\sigma=1.2$.

When compared with Fig.~\ref{fig:figure13}, it is apparent that the flow variance increases with increasing $\sigma$, per the intuition and theorem~\ref{meanandvariance}.
Also, it is observed that the deviation of flow average from the deterministic model increases as the stochastic noise becomes more vigorous, as predicted by corollary~\ref{lemTen}.
Moreover, the convex form in the average flow and the almond shape of the flow variance, inferred from corollary~\ref{lemMu} and theorem~\ref{meanandvariance}, become more pronounced for larger $\sigma$.

\begin{figure}[ht]
\begin{tabular}{ccc}
\begin{minipage}{170pt}
\centerline{\includegraphics[width=1\textwidth]{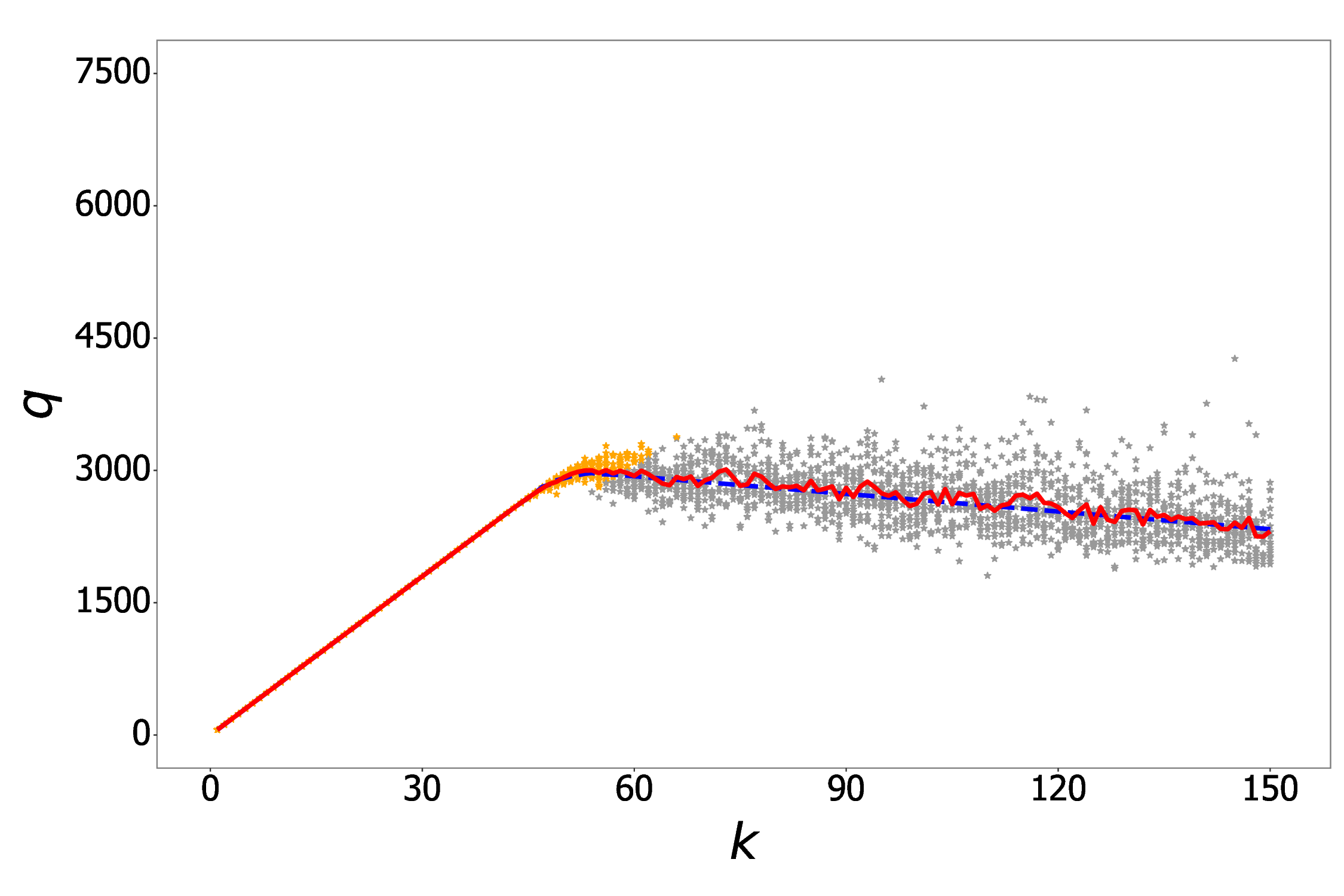}}
\end{minipage}
&
\begin{minipage}{170pt}
\centerline{\includegraphics[width=1\textwidth]{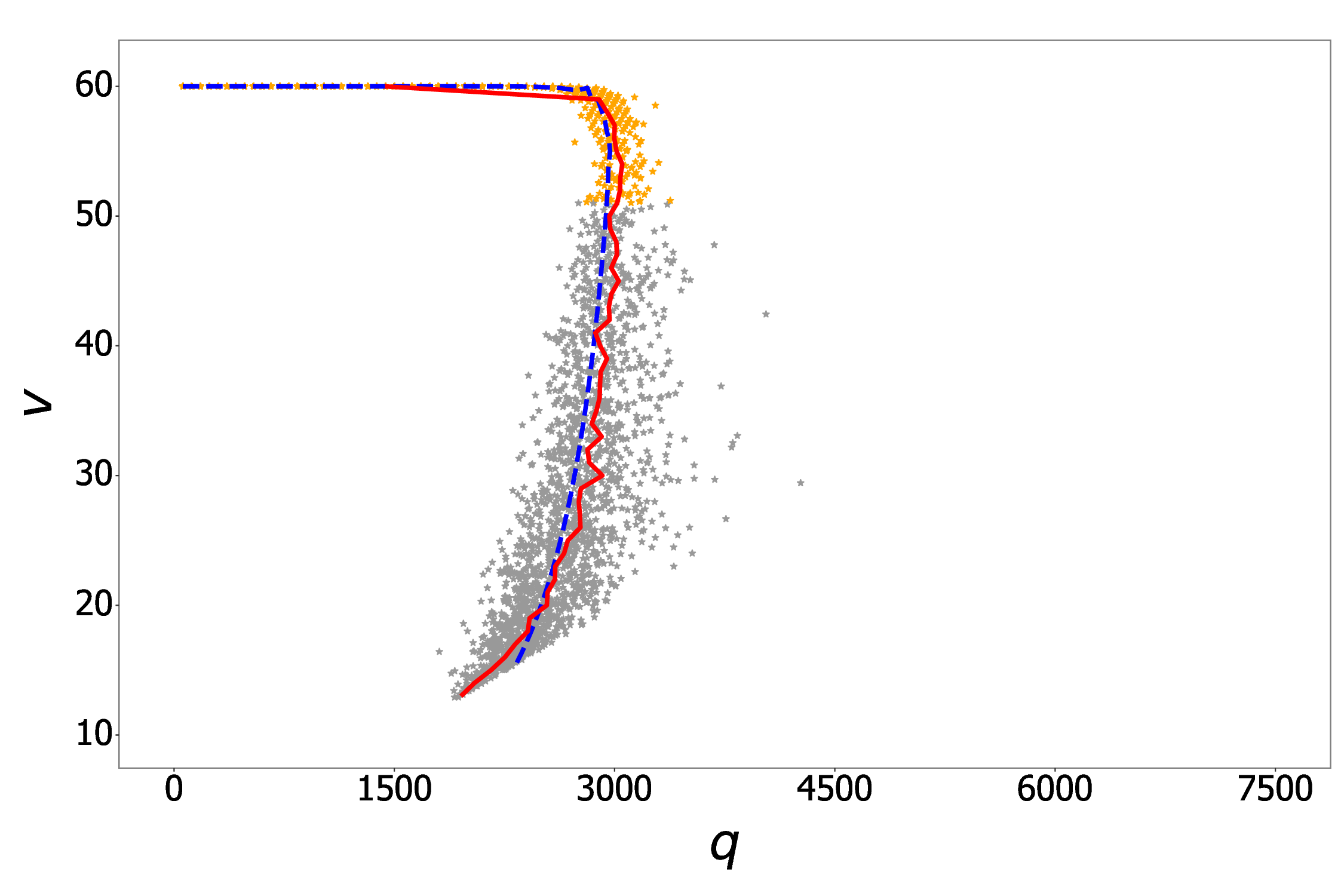}}
\end{minipage}
&
\begin{minipage}{170pt}
\centerline{\includegraphics[width=1\textwidth]{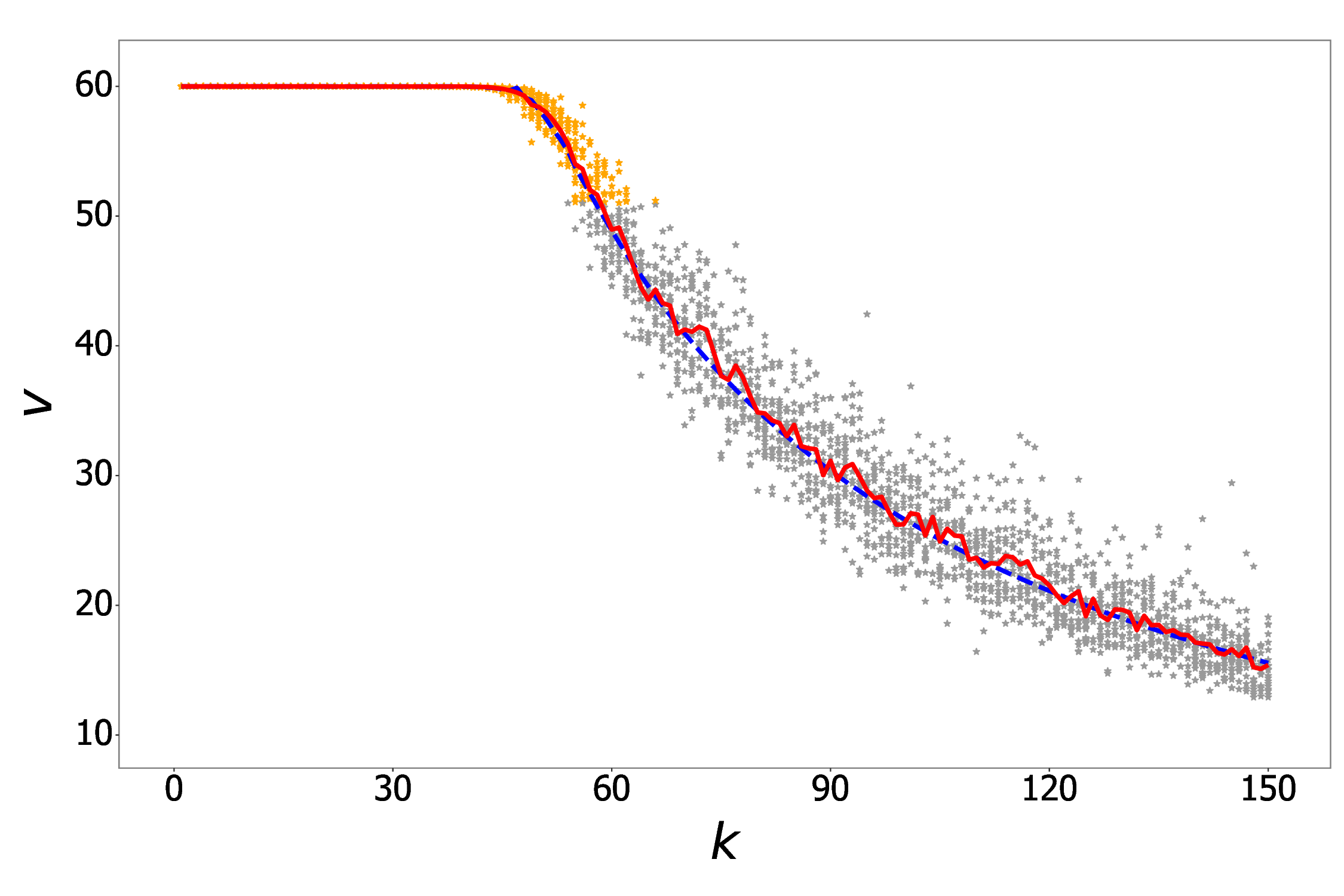}}
\end{minipage}\\
\begin{minipage}{170pt}
\centerline{\includegraphics[width=1\textwidth]{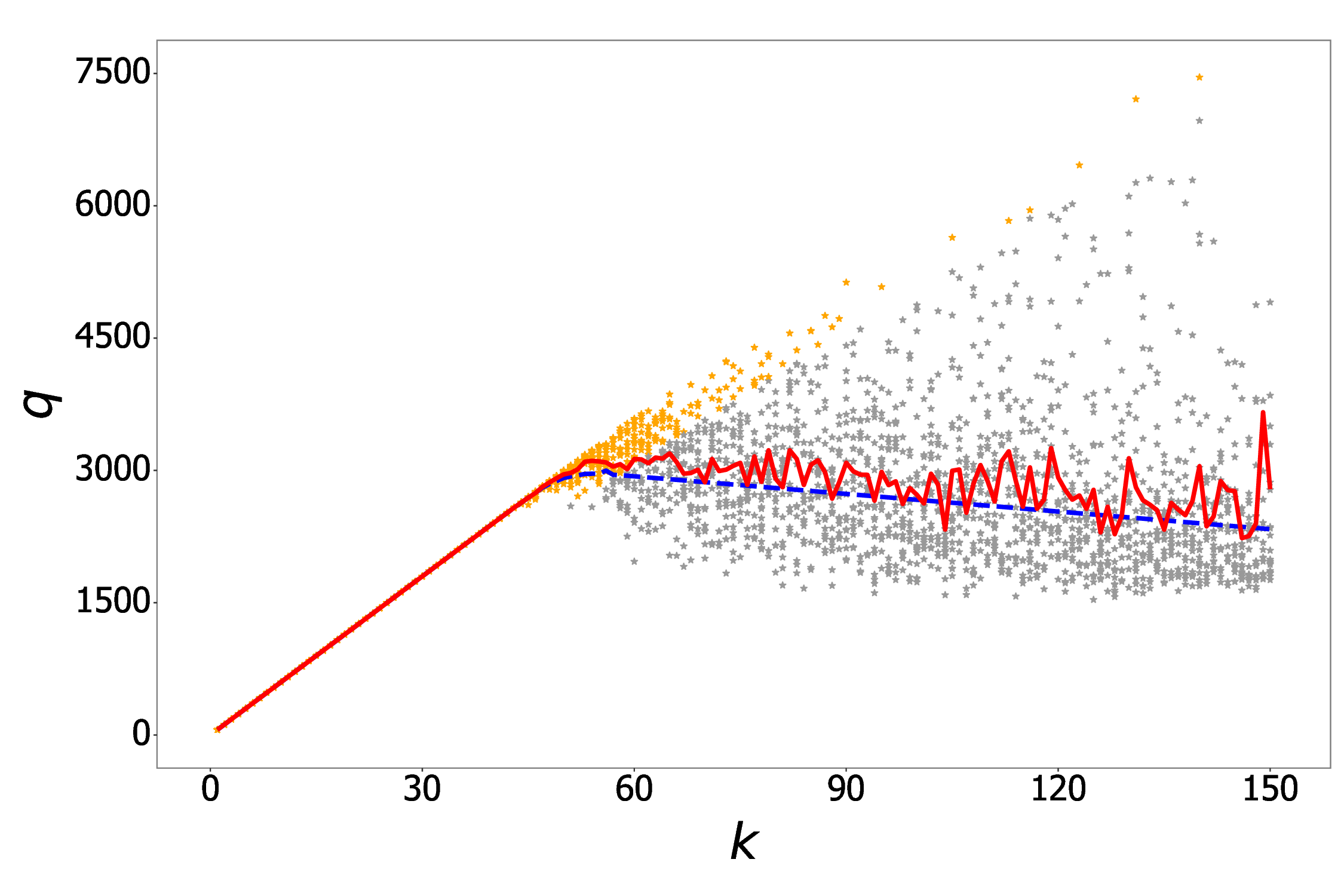}}
\end{minipage}
&
\begin{minipage}{170pt}
\centerline{\includegraphics[width=1\textwidth]{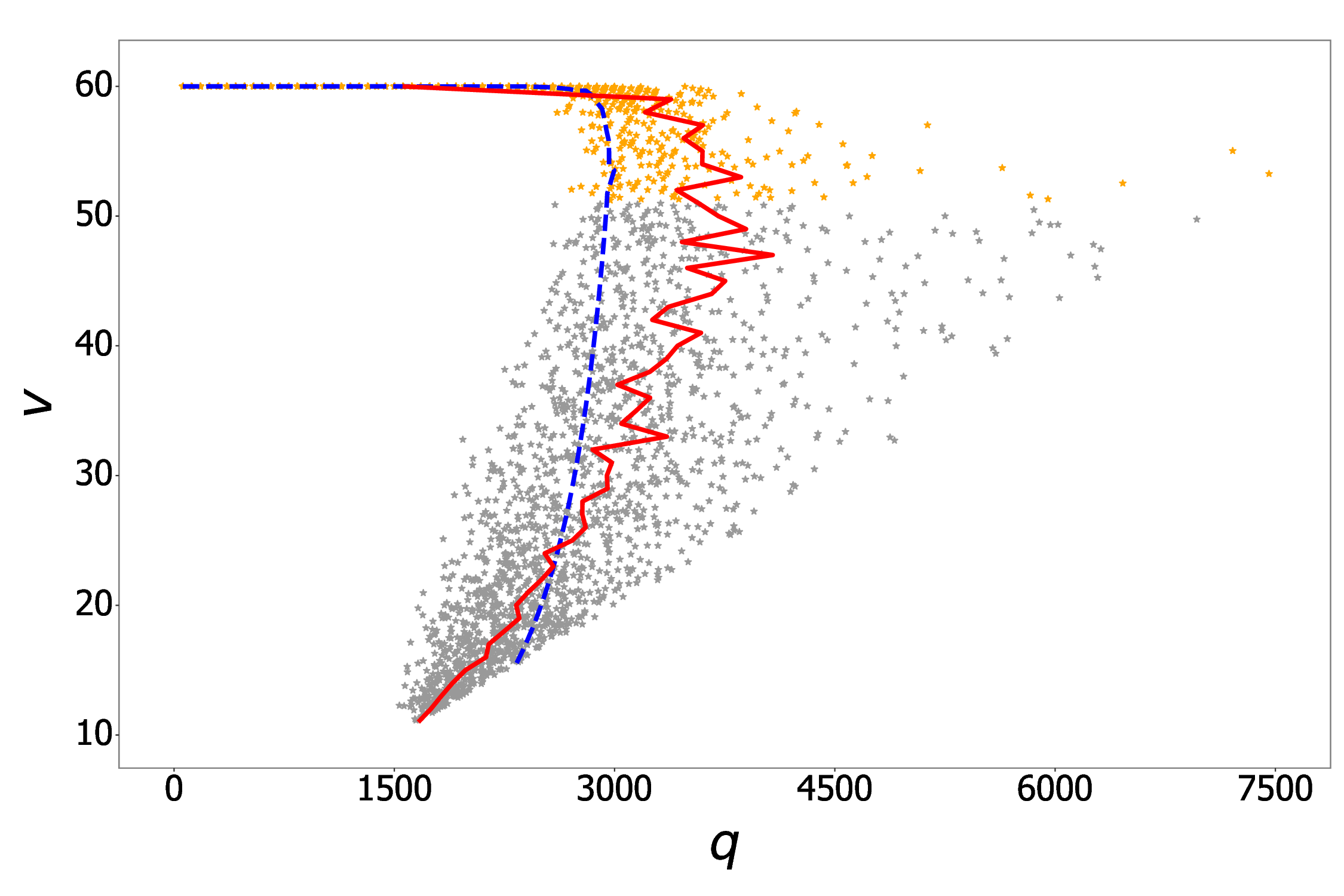}}
\end{minipage}
&
\begin{minipage}{170pt}
\centerline{\includegraphics[width=1\textwidth]{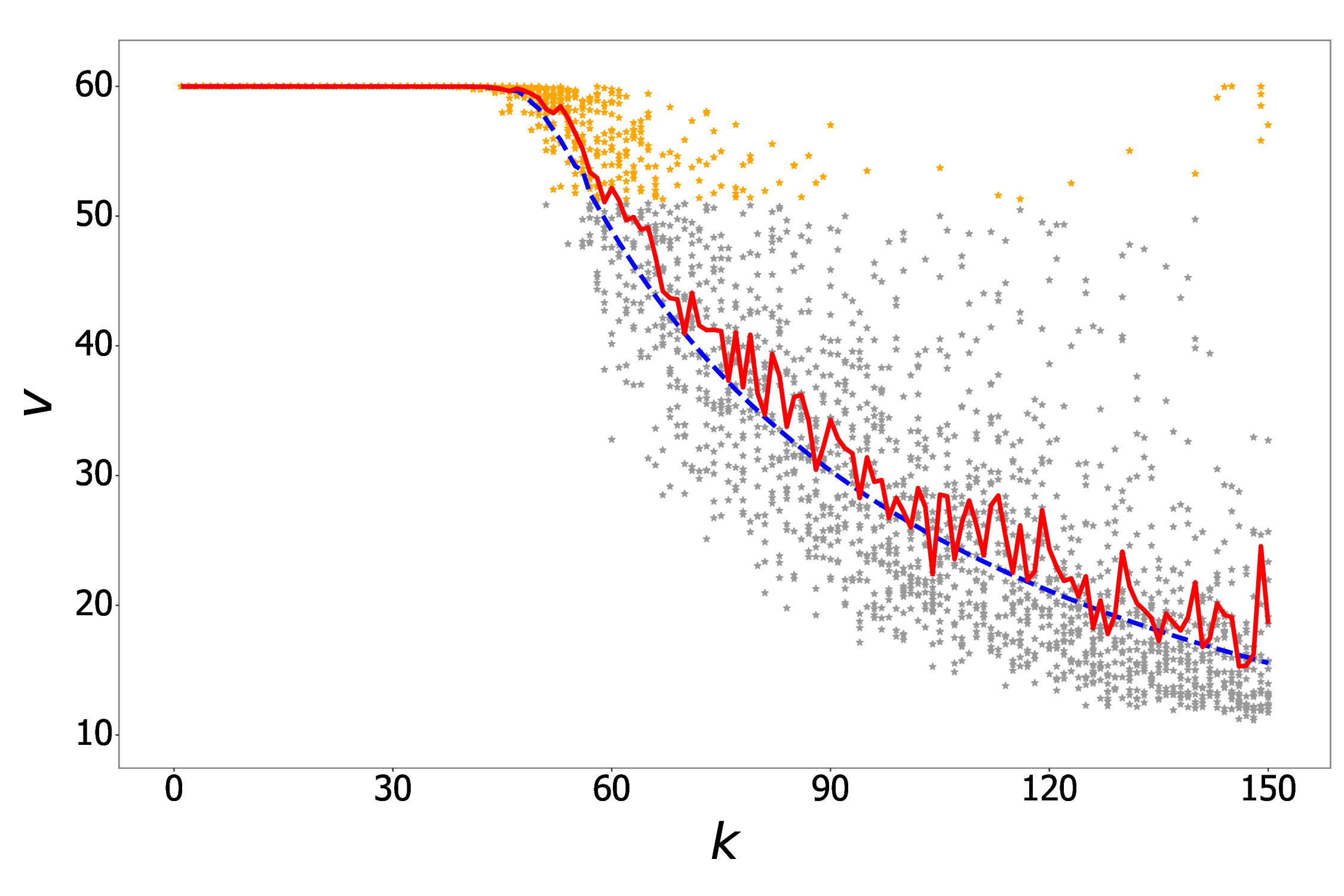}}
\end{minipage}
\end{tabular}\\
\renewcommand{\figurename}{Fig.}
\caption{The same as Fig.~\ref{fig:figure13} but considering $\sigma = 0.5$ (upper row) and $\sigma = 1.2$ (bottom row).
In the calculations, the following values are adopted for the remaining parameters: $v_1 = 10$, $v_2 = 60$, $L = 1$, $N_\mathrm{max} = 200$, $c_1 = 1$, $c_2 = 3$, $t_s= 26$, and $N_\mathrm{cut} = 150$.}
\label{fig:figure9}
\end{figure}

\bibliographystyle{h-physrev}
\bibliography{references_traffic, references_qian}

\end{document}